\documentclass[prd,aps,amsfonts,showpacs,longbibliography,notitlepage,twocolumn,groupedaddress,superscriptaddress]{revtex4-2}

\usepackage[utf8]{inputenc}

\usepackage{tikz} % Package for drawing
\usepackage{amsmath,amsfonts,amsthm, mathtools}
\usepackage{placeins}
\usepackage{float}
\usepackage{dsfont}
\usepackage{csquotes}
\usepackage{amssymb}
\usepackage{verbatim}
\usepackage{bm}
\usepackage{amsmath}
\usepackage{amssymb}
\usepackage{stmaryrd}
\usepackage{amsthm}
\usepackage[caption=false]{subfig}
\usepackage{physics}
\usepackage{bm}  

\usepackage{hyperref}
\definecolor{darkred}  {rgb}{0.5,0,0}
\definecolor{darkblue} {rgb}{0,0,0.5}
\definecolor{darkgreen}{rgb}{0,0.5,0}
\hypersetup{
  colorlinks = true,
  urlcolor  = blue,         % color of external links
  linkcolor = darkblue,     % color of internal links
  citecolor = darkgreen,    % color of links to bibliography
  filecolor = darkred       % color of file links
}

\theoremstyle{definition}

\newtheorem{thm}{Theorem}

\usepackage{multirow}

\definecolor{cool_green}{rgb}{0.0, 0.5, 0.0}

\newcommand{\yk}[1]{{#1}}

\begin{document}

%\setlength{\parindent}{0pt}
%\title{Quantifying performance of quantum imaging using learning theory}
\title{Limitations of Gaussian measurements in quantum imaging}
%\date{}
\author{Yunkai Wang}
\email{ywang10@perimeterinstitute.ca}
\affiliation{Perimeter Institute for Theoretical Physics, Waterloo, Ontario N2L 2Y5, Canada.}
\affiliation{Department of Applied Mathematics, University of Waterloo, Ontario N2L 3G1, Canada.}
\affiliation{Institute for Quantum Computing, University of Waterloo, Ontario N2L 3G1, Canada.}

\author{Sisi Zhou}
\affiliation{Perimeter Institute for Theoretical Physics, Waterloo, Ontario N2L 2Y5, Canada.}
\affiliation{Department of Physics and Astronomy, University of Waterloo, Ontario N2L 3G1, Canada.}
\affiliation{Institute for Quantum Computing, University of Waterloo, Ontario N2L 3G1, Canada.}
\affiliation{Department of Applied Mathematics, University of Waterloo, Ontario N2L 3G1, Canada.}

\begin{abstract}
% Imaging thermal sources naturally yields Gaussian states at the receiver, raising the question of whether imaging can be performed solely with Gaussian operations. In this work, we establish a no-go theorem that limits the performance of Gaussian measurements when imaging weak thermal sources with mean photon number per temporal mode $\epsilon \ll 1$. Our results show that non-Gaussian measurements can outperform any Gaussian measurement by a factor of $\epsilon$ in the variance of parameter estimation for both single-lens and interferometric imaging. We present several examples that confirm our general no-go theorem on this performance gap. As a secondary result, we also analyze the superresolution problem, further confirming our no-go theorem and demonstrating that superresolution is unachievable using only Gaussian measurements.
Imaging thermal sources naturally yields Gaussian states at the receiver, raising the question of whether Gaussian measurements can perform optimally in quantum imaging. \yk{In this work, we establish no-go theorems on the performance of Gaussian measurements for imaging thermal sources in the limit of mean photon number per temporal mode $\epsilon \to 0$ or source size $L \to 0$. We show that non-Gaussian measurements can outperform any Gaussian measurement in the scaling of the estimation variance with $\epsilon$ (or $L$).} We also present several examples to illustrate the no-go results. 
\end{abstract}

\maketitle
\textit{Introduction} - \yk{Imaging thermal sources is a crucial technique across a range of fields. With the rising interest in applying quantum technologies to imaging, increasing attention has been devoted to understanding the fundamental resolution limits in areas such as astronomy and microscopy \cite{tsang2016quantum,tsang2019resolving,zanforlin2022optical,parniak2018beating,tham2017beating,paur2016achieving,yu2018quantum,napoli2019towards,zhou2019modern,tsang2017subdiffraction,tsang2019quantum,wang2023fundamental,nair2016far,lupo2016ultimate,wang2021superresolution,xie2024far,yang2017fisher,dong2020superresolution,yang2016far,darji2024robust,bhusal2022smart,backlund2018fundamental,mitchell2024quantum,taylor2016quantum,defienne2024advances,taylor2014subdiffraction,taylor2013biological,casacio2021quantum,cameron2024adaptive,ndagano2022quantum,tenne2019super,he2023quantum,classen2017superresolution,jin2018nanoparticles,huang2009super,picariello2025quantum,yue2025quantum,kudyshev2023machine,gatto2014beating,rust2006sub,betzig2006imaging,gustafsson2000surpassing,hell1994breaking} and the potential benefits of leveraging quantum networks for interferometric imaging \cite{gottesman2012longer,khabiboulline2019quantum,khabiboulline2019optical,huang2022imaging,marchese2023large,czupryniak2022quantum,czupryniak2023optimal,wang2023astronomical,purvis2024practical,huang2024limited}.}
Notably, several studies have employed the formalism of Gaussian quantum information to analyze states emitted by thermal sources, rather than focusing solely on individual photons. For example, imaging resolution has been extensively analyzed for thermal sources of arbitrary strength \cite{nair2016far,lupo2016ultimate,wang2021superresolution,xie2024far,yang2017fisher,dong2020superresolution,yang2016far}. \yk{Heterodyne or homodyne detection in specific spatial modes has been investigated as a means to achieve superresolution, although such approaches do not succeed in attaining it \cite{xie2024far,yang2017fisher,dong2020superresolution,yang2016far}.} Meanwhile, interferometric imaging assisted by continuous-variable quantum networks has been explored, modeling stellar light as Gaussian states \cite{wang2023astronomical,purvis2024practical,huang2024limited}, revealing that homodyne detection, even with distributed entanglement, does not clearly outperform local measurement schemes \cite{purvis2024practical}.  Despite these advances in specific scenarios, a general theorem that unifies these discussions and clarifies the role of Gaussian measurements in imaging remains lacking. This work aims to fill this gap by establishing a no-go theorem.

We show that when imaging a weak thermal source with the mean photon number per temporal mode $\epsilon \ll 1$ in interferometric imaging, the Fisher information matrix (FIM) $F$ for estimating unknown parameters that inversely bounds the estimation variance \cite{kay1993fundamentals}, satisfies $\|F\|= NO(\epsilon^2)$ using any Gaussian measurement,  where $\|F\|$ is the largest eigenvalue of $F$ (spectrum norm),  $N$ is the number of copies of the measured state.   In contrast, non-Gaussian measurements can achieve $\|F\|= N\Theta(\epsilon)$~\footnote{In this work, \yk{$O(.)$ denotes asymptotic upper bound. $g(n)=O(f(n))$ if $\exists\,C>0,\,n_0$ such that $|g(n)|\le C|f(n)|$ for $n\ge n_0$. $\Theta(.)$ denotes asymptotically tight bound. $g(n)=\Theta(f(n))$ if $\exists\,C_{1,2}>0,\,n_0$ such that $C_1|f(n)|\le |g(n)|\le C_2|f(n)|$ for $n\ge n_0$.}}. 
Interestingly, this performance gap mirrors the gap between local and nonlocal measurements in interferometric imaging, where any local measurement is limited to $\|F\| =N O(\epsilon^2)$, while nonlocal measurements can achieve $\|F\| = N\Theta(\epsilon)$, as shown in Ref.~\cite{tsang2011quantum}. %\yk{Here, a local measurement refers to any measurement implementable using local operations and classical communication (LOCC), while a nonlocal measurement is one that cannot be realized using LOCC. } 
This stark difference between local and nonlocal measurements has motivated extensive discussions on interferometric imaging assisted by quantum networks \cite{gottesman2012longer,khabiboulline2019quantum,khabiboulline2019optical,huang2022imaging,marchese2023large,czupryniak2022quantum,czupryniak2023optimal,wang2023astronomical,purvis2024practical,huang2024limited}.
Our work complements Ref.~\cite{tsang2011quantum} by demonstrating that nonlocality alone is insufficient to achieve the improved scaling of $\|F\| = N\Theta(\epsilon)$. We establish that non-Gaussianity is also a necessary condition for this enhanced performance. Specifically, even nonlocal Gaussian measurements are fundamentally constrained to at most $\|F\| = NO(\epsilon^2)$.
Moreover, our results extend to general parameter estimation in single-lens imaging of weak thermal sources, which measures the light field formed on the detection plane by a single lens, where the concept of locality is not applicable.
We prove that the same significant performance gap between Gaussian and non-Gaussian measurements persists.

\yk{We also investigate the superresolution problem \cite{tsang2016quantum,tsang2019resolving,zanforlin2022optical,parniak2018beating,tham2017beating,paur2016achieving,yu2018quantum,napoli2019towards,zhou2019modern,tsang2017subdiffraction,tsang2019quantum,wang2023fundamental,nair2016far,lupo2016ultimate,wang2021superresolution,xie2024far,yang2017fisher,dong2020superresolution,yang2016far,darji2024robust,bhusal2022smart}, which concerns imaging a source of size $L$ far below the Rayleigh limit. Previous work has shown that a suitably designed non-Gaussian measurement—specifically, photon counting in Hermite–Gaussian spatial modes—can achieve a FIM $F$ with more favorable scaling in $L$ than direct imaging \cite{tsang2016quantum,zhou2019modern,tsang2017subdiffraction,tsang2019quantum}. This enhanced scaling enables significantly improved performance in the sub-Rayleigh regime, a phenomenon known as superresolution. In contrast, we prove that any Gaussian measurement cannot surpass direct imaging in terms of the FIM scaling with $L$ in this regime. Our results thus establish a new no-go theorem, providing fundamental insights for the design of quantum imaging protocols.

}

From a broader perspective, the role of Gaussian and non-Gaussian operations has been extensively studied in quantum information theory. In the context of quantum computing with Gaussian states, it has been shown that non-Gaussian measurements are essential for achieving universal quantum computation \cite{lloyd1999quantum,bartlett2002efficient,weedbrook2012gaussian}. Non-Gaussianity has also been recognized as a valuable resource in quantum resource theory \cite{chitambar2019quantum} and has been shown to be essential for continuous-variable entanglement distillation \cite{eisert2002distilling}, quantum error correction \cite{niset2009no}, and various other applications.
However, some quantum information tasks, such as quantum key distribution \cite{grosshans2002continuous,diamanti2015distributing,cerf2007quantum} and quantum teleportation \cite{braunstein1998teleportation}, can be implemented solely with Gaussian operations and states.
Whether Gaussian measurements are sufficient for quantum sensing tasks remains an open question. Our work introduces a new no-go theorem for Gaussian operations in the context of quantum imaging, demonstrating that non-Gaussianity is also a critical resource in quantum sensing. \yk{Specifically, the absence of non-Gaussianity can significantly degrade the performance of imaging weak thermal sources or superresolution techniques.}

\begin{figure}[!tb]
\begin{center}
\includegraphics[width=0.8\columnwidth]{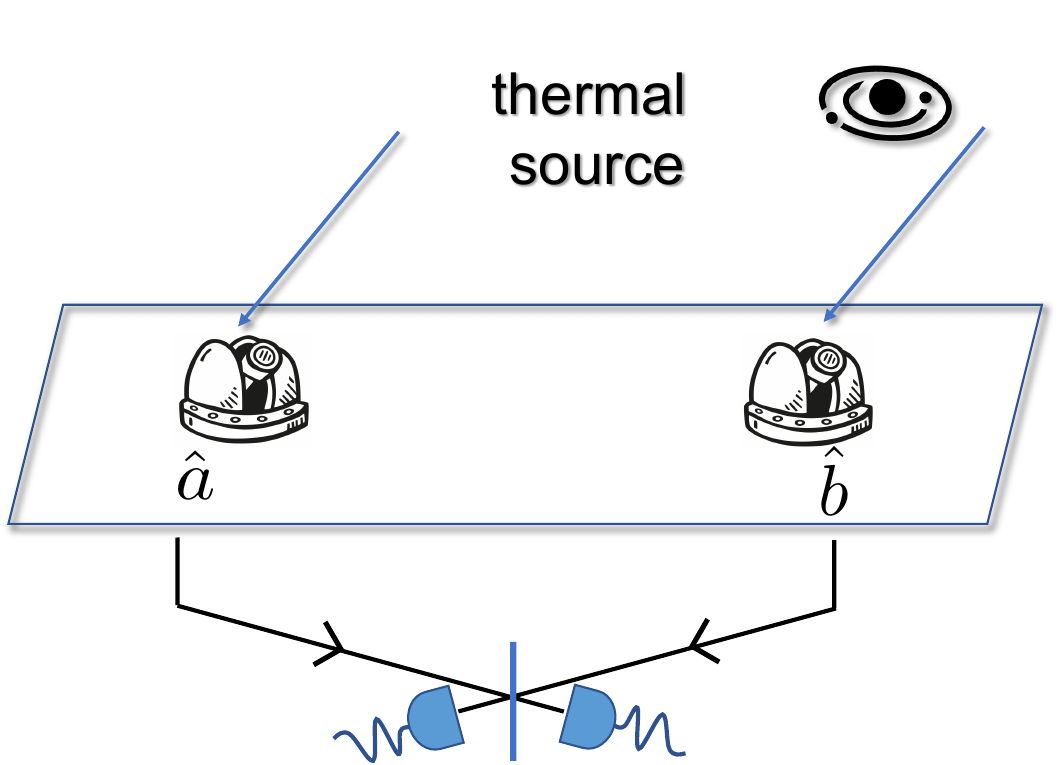}
\caption{Set up for interferometric imaging with two lenses corresponding to the two spatial modes $\hat{a},\hat{b}$.  }
\label{fig:inteferometer}
\end{center}
\end{figure}

\yk{\textit{Imaging weak thermal sources} - We begin by discussing the imaging of a weak thermal source in the simplest scenario: interferometric imaging with two lenses, as shown in Fig.~\ref{fig:inteferometer}. This setup involves only two spatial modes, making the proof more straightforward than in the general case. We then extend the analysis to interferometric imaging with more lenses and to single-lens imaging.}
Interferometric imaging utilizes multiple lenses to function collectively as a larger effective lens, where the diameter of this synthetic lens is determined by the baseline between the smaller lenses  \cite{monnier2003optical}. According to the van Cittert-Zernike theorem \cite{zernike1938concept}, the mutual coherence function of light collected by different lenses corresponds to a Fourier component of the source’s intensity distribution, with its spatial frequency determined by the baseline.
The two-mode weak thermal state received by the two lenses can be described using the Sudarshan-Glauber P representation \cite{mandel1995optical}
\begin{equation}\begin{aligned}\label{rho_thermal}
&\rho=\int\frac{d^2\alpha d^2\beta}{\pi^2\det\Gamma}\exp(-\vec{\gamma}^\dagger \Gamma^{-1}\vec{\gamma})\ket{\vec{\gamma}}\bra{\vec{\gamma}},\\
&\vec{\gamma}=[\alpha,\beta]^T,\quad \Gamma=  \frac{\epsilon}{2}\left[
\begin{matrix}
1 & g\\
g^* & 1
\end{matrix}\right],\\
&\ket{\vec{\gamma}}=\exp(\alpha \hat{a}^\dagger-\alpha^*\hat{a})\exp(\beta \hat{b}^\dagger-\beta^*\hat{b})\ket{0},
\end{aligned}\end{equation}
where $g = |g|e^{i\theta}$ represents the coherence function, and $\epsilon$ denotes the mean photon number per temporal mode, which is assumed to be much less than one ($\epsilon \ll 1$). The operators $\hat{a}$ and $\hat{b}$ are the annihilation operators for the two modes. \yk{As a Gaussian state whose Wigner function has a Gaussian form,} $\rho$ is fully characterized by its displacement $\mu_i=\langle \hat{z}_i\rangle$ and covariance matrix $V_{ij}=\frac{1}{2}\langle \{\hat{z}_i-\mu_i,\hat{z}_j-\mu_j\}\rangle$ \cite{weedbrook2012gaussian}, where $\hat{\vec{z}}=[\hat{x}_1,\hat{p}_1,\hat{x}_2,\hat{p}_2]$, $\hat{a}=(\hat{x}_1+i\hat{p}_1)/\sqrt{2}$, $\hat{b}=(\hat{x}_2+i\hat{p}_2)/\sqrt{2}$, $\langle\hat{O}\rangle=\tr(\rho\hat{O})$, $\{\hat{O}_1,\hat{O}_2\}=\hat{O}_1\hat{O}_2+\hat{O}_2\hat{O}_1$. The displacement vanishes, and the covariance matrix is given by
\begin{equation}
V_\rho=\frac{1}{2}\left[
\begin{matrix}
1+\epsilon & 0 & \epsilon |g|\cos\theta & -\epsilon|g|\sin\theta\\    
0 & 1+\epsilon & \epsilon|g|\sin\theta & \epsilon|g|\cos\theta\\
\epsilon|g|\cos\theta & \epsilon|g|\sin\theta & 1+\epsilon & 0\\
-\epsilon|g|\sin\theta & \epsilon|g|\cos\theta & 0 & 1+\epsilon
\end{matrix}\right].
\end{equation}
%A nonlocal non-Gaussian measurement that projects onto $(\ket{01} + e^{i\delta} \ket{10}) / \sqrt{2}$ yields a Fisher information of $F = \Theta(N\epsilon)$ when measuring $N$ copies of $\rho$, as demonstrated in Sec.~\ref{appendix:non-Gaussian} of the Supplemental Material. 
\yk{We compare the performance of Gaussian and non-Gaussian measurements on this state for imaging. A measurement is  called Gaussian if, when applied to any Gaussian state, it yields outcomes whose probability distribution is also Gaussian   \cite{weedbrook2012gaussian}. Experimentally, any Gaussian measurement can be implemented using homodyne detection, beam splitters, squeezers, displacements, and ancilla modes prepared in Gaussian states. Measurements not realizable
in this way are non-Gaussian. If we perform a non-Gaussian, nonlocal measurement—photon number detection at the two output ports of Fig.~\ref{fig:inteferometer}—the FIM for estimating the unknown parameters $|g|$ and $\theta$ achieves $\|F\| = N\Theta(\epsilon)$ \cite{tsang2011quantum}.} The FIM lower bounds the covariance matrix of estimating a set of parameters $\vec{x}$ through the Cramér-Rao bound, which states that $\text{Cov}(\hat{\vec{x}}) \geq F^{-1}$ (i.e., $\text{Cov}(\hat{\vec{x}}) - F^{-1}$ is positive semidefinite) for any unbiased estimator $\hat{\vec{x}}$, and the bound is asymptotically saturable by the maximum likelihood estimator under proper regularity conditions \cite{kay1993fundamentals}. \yk{We now establish our first theorem which upper bounds the FIM for any Gaussian measurement.} %, including nonlocal Gaussian measurements. 

%Note that in this comparison, we assume the implementation is ideal without transmission loss or noise and hence represents a fundamental performance difference.

\begin{thm}\label{thm:two_lens}
For interferometric imaging with two lenses that    receive $N$ copies of states in the form given in Eq.~\ref{rho_thermal}, each element of the FIM for estimating the unknown parameters $\theta$ and $|g|$ using any Gaussian measurement is upper bounded by
\begin{equation}\begin{aligned}
&F_{|g||g|}\leq 2\epsilon^2N, \\  
&F_{\theta\theta}\leq 2\epsilon^2|g|^2N, \\  
&F_{\theta|g|}\leq 2\epsilon^2|g|N. \\  
\end{aligned}
\end{equation}
\end{thm}
\begin{proof}
Any Gaussian measurement can be written as the form \cite{adesso2014continuous,weedbrook2012gaussian}
\begin{equation}\label{Pi_y}
\Pi_{\vec{y}}=\frac{1}{\pi^2}D_{\vec{y} }\Pi_0D^\dagger_{\vec{y}},
\end{equation}
where $\Pi_0$ is the density matrix of a general Gaussian state with vanishing displacement and covariance matrix $V_\Pi$. Notably, the outcome label $\vec{y}$ is solely determined by the displacement of $\Pi_{\vec{y}}$. As demonstrated in Sec.~\ref{appendix:Gaussian_measurement} of the Supplemental Material, the probability distribution is
\begin{equation}\begin{aligned}
&P(\vec{y}|g)=\frac{1}{(2\pi)^2\sqrt{\det V}}\exp[-\frac{1}{2}\vec{y}^TV^{-1}\vec{y}],
\end{aligned}\end{equation}
where $V = V_\Pi + V_\rho$. The FIM $F$ for estimating the unknown parameters $\vec{x}$ can then be computed from a Gaussian probability distribution  $\vec{y} \sim \mathcal{N}(\vec{0}, C(\vec{x}))$ as \cite{kay1993fundamentals}
\begin{equation}
\begin{aligned}
[F]_{ij}&=\frac{1}{2}\tr[C^{-1}(\vec{x})\frac{\partial C(\vec{x})}{\partial x_i}C^{-1}(\vec{x})\frac{\partial C(\vec{x})}{\partial x_j}].
\end{aligned}
\end{equation}
We can then calculate the FIM of estimating $|g|,\theta$.
Define
\begin{equation}
U_1=\frac{1}{\sqrt{2}}\left[
\begin{matrix}
\sin\theta & -\cos\theta & 0 & 1\\    
-\cos\theta & -\sin\theta & 1 & 0\\
-\sin\theta & \cos\theta & 0 & 1\\
\cos\theta & \sin\theta & 1 & 0
\end{matrix}\right].
\end{equation}
We can define $\Sigma=U_1VU_1^\dagger$ and 
\begin{equation}
\begin{aligned}
&\Sigma_\rho=U_1V_\rho U_1^\dagger=\frac{1}{2}\text{diag}[a,a,b,b],\\
&\Sigma_{\partial|g|}=U_1\frac{\partial V}{\partial |g|} U_1^\dagger=\frac{1}{2}\text{diag}[-\epsilon,-\epsilon,\epsilon,\epsilon],
\end{aligned}
\end{equation}
where $a=1+\epsilon-\epsilon|g|,b=1+\epsilon+\epsilon|g|$. 
If we consider $N$ copies of the state $\rho^{\otimes N}$, its covariance matrix is given by $V_\rho^N = I_N \otimes V$, \yk{where $I_N$ is the $N$-dimensional identity matrix}, but we allow $V_\Pi^N$ to be general, rather than having this tensor product structure. 
We also define $V^N = V_\rho^N + V_\Pi^N$ and $\Sigma^N = (I_N \otimes U_1)V^N(I_N \otimes U_1)^\dagger$, with similar definitions for $\Sigma^N_\Pi$ and $\Sigma^N_\rho$.
We can find the FIM element of estimating $|g|$ 
\begin{equation}
\begin{aligned}
F_{|g||g|}&=\frac{1}{2}\tr((\Sigma^N)^{-1}\Sigma^N_{\partial|g|}(\Sigma^N)^{-1}\Sigma^N_{\partial|g|})\leq \frac{\epsilon^2}{8}\tr((\Sigma^N)^{-2})\\
&\leq  \frac{\epsilon^2}{8}\tr((\Sigma_\rho^N)^{-2})\leq2\epsilon^2 N,
\end{aligned}\end{equation}
where we take the absolute value of the eigenvalues of $\Sigma^N_{\partial|g|}$ in the first inequality because, using the spectral decomposition $\Sigma^N_{\partial|g|} = \sum_i \lambda_i \ket{v_i} \bra{v_i}$, we obtain $F_{|g||g|}=\frac{1}{2}\sum_{i,j}\lambda_i\lambda_j|\bra{v_i}(\Sigma^N)^{-1}\ket{v_j}|^2\leq \frac{1}{2}\sum_{i,j}|\lambda_i\lambda_j||\bra{v_i}(\Sigma^N)^{-1}\ket{v_j}|^2$. In the second inequality, we use the fact that $\Sigma_\Pi^N$ is positive semidefinite.  Similarly, we can establish the upper bound for $F_{\theta\theta}$ and $F_{\theta|g|}$. %, as detailed in Sec.~\ref{appendix:proof thm interferometric} of the Supplemental Material.
\end{proof}
\yk{We have thus shown that any Gaussian measurements always yield a FIM scaling of at most $NO(\epsilon^2)$. Notably, the Gaussian measurements considered in the above proof also include nonlocal Gaussian measurements.
Ref.~\cite{tsang2011quantum} demonstrates a scaling difference between the nonlocal and local measurement. Theorem~\ref{thm:two_lens} extends this analysis by revealing the same scaling gap between the Gaussian and non-Gaussian measurement.} Consequently, our findings clearly quantify the limitations of relying solely on Gaussian measurements. Note that in the above proof, we allow the Gaussian measurement to be performed jointly on the $N$ copies of the state $\rho^{\otimes N}$. This implies that even a joint Gaussian measurement on multiple copies of the state still achieves only the $NO(\epsilon^2)$ scaling. \yk{Moreover, while we have bounded each element of the FIM individually, we can also obtain a straightforward upper bound on the entire matrix using the matrix inequality $F \leq 2N\epsilon^2(1+|g|^2) I_2$, where $I_2$ is the $2$-dimensional identity matrix.  Having established the case of interferometric imaging with two lenses, we also extend our analysis to the general scenario with multiple lenses. Our results show that any Gaussian measurement still remains limited by the same performance bound, $\|F\| = NO(\epsilon^2)$. 
} 

% , as shown in Sec.~\ref{appendix:proof thm interferometric} of the Supplemental Material. 
%Our bound may not be tight, but as we will see in the following examples, the scaling $NO(\epsilon^2)$ is attainable in some scenarios.

\yk{We now present several examples that illustrate the above no-go theorem.}
First, we consider the case where the light received by two lenses is directly combined on a beam splitter, followed by homodyne or heterodyne detection at the two output ports, as illustrated in Fig.~\ref{fig:inteferometer}. %In Sec.~\ref{appendix:examples} of the Supplemental Material, 
\yk{We explicitly calculate the FIM and show that it scales as $\|F\|= N\Theta(\epsilon^2) $ across different scenarios of homodyne and heterodyne detection, confirming the general theorem stated above.  }
Furthermore, interferometric imaging based on continuous-variable quantum teleportation can achieve FIM scaling of $\|F\| = N\Theta(\epsilon)$ with non-Gaussian photon counting detection \cite{wang2023astronomical}, whereas it is limited to $\|F\| = NO(\epsilon^2)$ if only homodyne detection is used~\cite{purvis2024practical}.  These discussions also serve as specific examples that align with our general no-go results, highlighting that achieving the scaling $\|F\|=N\Theta(\epsilon)$ with distributed entanglement still requires non-Gaussian measurement.

\yk{For single-lens imaging as shown in Fig.~\ref{fig:single lens}, where the distinction between nonlocal and local measurements in interferometric imaging, as discussed in Ref.~\cite{tsang2011quantum}, does not apply, the difference between Gaussian and non-Gaussian measurements remains relevant.
Interestingly, for single-lens imaging, we observe the same scaling difference, with non-Gaussian measurement achieving $\|F\| = N\Theta(\epsilon)$ and Gaussian measurement limited to $\|F\| = NO(\epsilon^2)$ for the estimation of an set of unknown parameters  encoded in the intensity distribution of the source. 
Note that the definition of Gaussian measurements is independent of the detector's spatial mode profile;  a detector with a Gaussian beam shape can still perform a non-Gaussian measurement if it involves, for example, photon number detection.
For detailed proofs of the no-go theorems for imaging weak thermal sources with Gaussian measurements in both interferometric and single-lens imaging, as well as additional details of the examples, we refer the reader to Sec.~\ref{appendix:imaging_weak_source} of the Supplemental Material.} 

%$\vec{\theta} = [\theta_1, \theta_2, \dots, \theta_Q]$
\begin{figure}[!tb]
\begin{center}
\includegraphics[width=1\columnwidth]{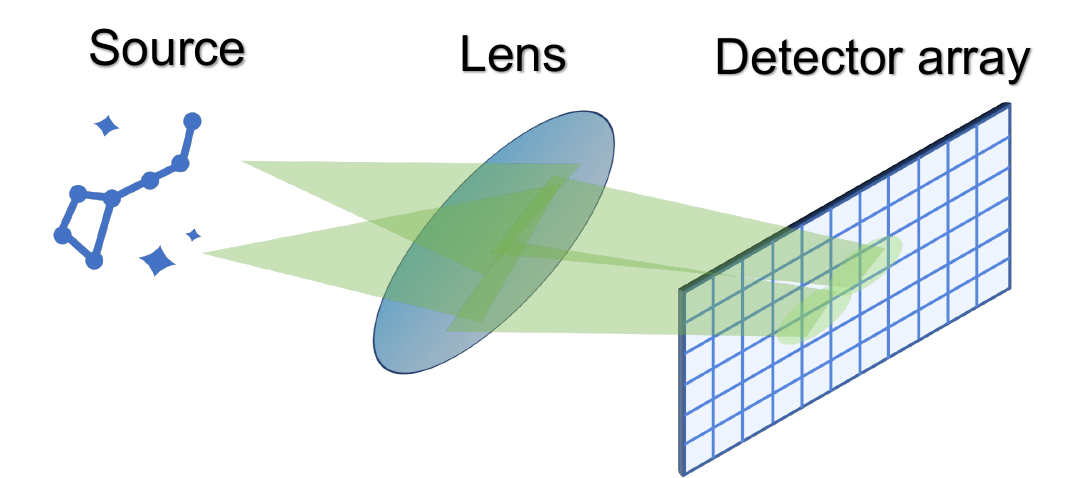}
\caption{Set up for single lens imaging.}
\label{fig:single lens}
\end{center}
\end{figure}

\textit{Superresolution} - \yk{We now explore the superresolution problem  as studied extensively in Ref.~\cite{tsang2016quantum,tsang2019resolving,zanforlin2022optical,parniak2018beating,tham2017beating,paur2016achieving,yu2018quantum,napoli2019towards,zhou2019modern,tsang2017subdiffraction,tsang2019quantum,wang2023fundamental,nair2016far,lupo2016ultimate,wang2021superresolution,xie2024far,yang2017fisher,dong2020superresolution,yang2016far,darji2024robust}. For the single-lens imaging,  the  thermal state received on the detection plane from a general incoherent source after a single lens has the form %(see Sec.~\ref{appendix:proof thm single lens} of the Supplemental Material for detailed justification)
\begin{equation}\begin{aligned}\label{eq:rho_single_lens}
&\rho_W = \int \frac{d^{2W}\vec{\gamma}}{\pi^W \det \Gamma}\,\exp\Bigl[-\vec{\gamma}^\dagger \Gamma^{-1}\vec{\gamma}\Bigr] \ket{\vec{\gamma}}\bra{\vec{\gamma}}, \\ 
&\ket{\vec{\gamma}}=\exp(\sum_i\gamma_{x_i}c_{x_i}^\dagger-\gamma_{x_i}^*c_{x_i})\ket{0},\\
&\vec{\gamma}=[\gamma_{x_1},\gamma_{x_2},\cdots,\gamma_{x_W}]^T,\\
&\Gamma=\sum_{i=1}^Q\epsilon \zeta_i\psi_i\psi_i^\dagger,\quad \sum_{i=1}^Q \zeta_i=1,
%&\Gamma_{jk} = \epsilon\sum_i J_i\psi_{ji}\psi_{ki}^*,
\end{aligned}
\end{equation}
where we consider $W$ points on the detection plane and $Q$ points on the source plane, $c_x$ is the annihilation operator at position $x$ on the detection plane, $\psi_i=[\psi(x_1-y_i),\psi(x_2-y_i),\cdots,\psi(x_W-y_i)]^T$, with $x_i$ denoting the position of the $i$th point on the detection plane, $y_i$ denoting the position of the $i$th point on the source plane, and we take $W,Q\to\infty$ in our derivation, $\psi(x)$ is the point spread function (PSF) and can take any form, $\zeta_i$ is the normalized intensity of the $i$th point on the source plane.
The total intensity is given by $\tr(\Gamma)=\epsilon$.
All points on the source plane with nonvanishing intensity $\zeta_i\neq0$ have positions confined to $y_i \in [y_0 - L/2,, y_0 + L/2]$, where $y_0$ is the source centroid and the source size $L \to 0$.

Previous work   has shown that for estimating the moments defined as $t_n=\sum_{i=1}^Q \zeta_i\left(\frac{y_i-y_0}{L}\right)^n$, $n=0,1,2,\cdots$,  which form a complete set of parameters characterizing the intensity distribution $\zeta_i$, a superresolution scheme based on non-Gaussian measurements can achieve FIM scaling as $F_{t_i t_j}=\Theta(L^{i+j-2\lfloor \min{i,j}/2 \rfloor})$ \cite{zhou2019modern,tsang2017subdiffraction,tsang2019quantum}. In contrast, for direct imaging the FIM scales as $F_{t_i t_j}=\Theta(L^{i+j})$. In the following, we show that when restricted to Gaussian measurements, the FIM is always bounded by $F_{t_i t_j}=O(L^{i+j})$, implying that one can never surpass direct imaging in terms of scaling with $L$.

% \begin{thm}\label{thm:superreoslution}
% If we consider imaging two thermal point sources in one dimension at positions $\pm L/2$ with equal strength using a single lens that receives a state of the form given in Eq.~\ref{eq:rho_single_lens} with
% \begin{equation}\begin{aligned}
% &\Gamma=\frac{\epsilon}{2}(\vec{\psi}_0\vec{\psi}_0^T+\vec{\psi}_1\vec{\psi}_1^T),\\
% \end{aligned}
% \end{equation}
% where $\psi_{0,1}=[\psi(x_1\pm L/2),\psi(x_2\pm L/2),\dots,\psi(x_W\pm L/2)]^T$, with $x_i$ denoting the position of the $i$th point on the detection plane, and we take $W\to\infty$ in our derivation, $\psi(x)=(2\pi\sigma^2)^{-1/4}\exp(-x^2/(4\sigma^2))$ is the PSF. Our goal is to estimate $L$ by measuring $N$ copies of states using Gaussian measurement. The FI is bounded by
% \begin{equation}
% F_{LL} = N \epsilon^2 O( L^2).
% \end{equation}

% \end{thm}

\begin{thm}\label{thm:superreoslution}
Consider imaging thermal sources of size $L \to 0$ with a single lens that receives a state of the form given in Eq.~\ref{eq:rho_single_lens}. For estimating the moments $\{t_n\}_{n=0,1,2,\cdots}$ by performing Gaussian measurements on $N$ copies of the state, the FIM is bounded by
\begin{equation}
F_{t_nt_m} = N \epsilon^2 O( L^{n+m}).
\end{equation}

\end{thm}

}

%We refer the reader to Sec.~\ref{appendix:proof thm superresolution} of the Supplemental Material for the derivation of the state and a detailed proof.  We find that FI cannot be independent of $L$, implying that superresolution, as in Ref.~\cite{tsang2016quantum,tsang2019resolving,zanforlin2022optical,parniak2018beating,tham2017beating,paur2016achieving,yu2018quantum,napoli2019towards,zhou2019modern,tsang2017subdiffraction,tsang2019quantum,wang2023fundamental,nair2016far,lupo2016ultimate,wang2021superresolution,xie2024far,yang2017fisher,dong2020superresolution,yang2016far}, is fundamentally unachievable solely with Gaussian measurements. 

%Theorem~\ref{thm:superreoslution} can also be extended to the case of interferometric imaging, demonstrating that any Gaussian measurement can achieve at most an FI of order $N\epsilon^2 O(L^2)$, as detailed in Sec.~\ref{appendix:proof thm superresolution}. Consequently, any Gaussian measurements cannot achieve superresolution for interferometric imaging either.

\yk{
Note that Theorem~\ref{thm:superreoslution} applies to a thermal source with any $\epsilon$, without requiring $\epsilon \to 0$, meaning that even for a strong thermal source, Gaussian measurement still cannot achieve superresolution. 
Furthermore, the FI is bounded by at most $NO(\epsilon^2)$ when using Gaussian measurements, in agreement with the no-go results in the weak thermal source limit. However, Theorem~\ref{thm:superreoslution} further reveals important scaling properties with respect to $L$, which are relevant to superresolution. Similar no-go results for the superresolution of resolving the distance between two point sources is also established and can actually be regarded as a special case of the above general theorem.  Note that superresolution methods typically involve measurements in specific spatial modes. For example, Ref.~\cite{tsang2016quantum} shows that photon counting in Hermite–Gaussian spatial modes can achieve superresolution for a Gaussian PSF. We emphasize that a measurement is classified as Gaussian or non-Gaussian based on whether the probability distribution resulting from measuring a Gaussian state is itself Gaussian, regardless of the spatial mode’s shape. For example, photon counting in Hermite–Gaussian spatial modes constitutes a non-Gaussian measurement.
Our general no-go theorem applies to any Gaussian measurements performed in any spatial modes and shows that such measurements still cannot achieve superresolution.

In deriving Theorem \ref{thm:superreoslution}, we do not make any explicit assumption about the shape of the PSF; therefore, the results hold for an arbitrary PSF. The shape of the PSF could affect the prefactor of $F_{t_nt_m}$ while scaling over $L$ remains the same. Note that the discussion based on FIM and Cramer-Rao bound provides the precision quantification only for the unbiased estimator, we avoid such a loophole following the approach in Ref.~\cite{tsang2018conservative}. We show in Sec.~\ref{Appendix:Unbiased estimator} of the Supplemental Material that, even when biased estimators are considered, Gaussian measurements still fail to achieve a substantial improvement over direct imaging in resolving the separation between two point sources.

Previous studies \cite{xie2024far,yang2017fisher,dong2020superresolution,yang2016far} have explored  superresolution through heterodyne or homodyne detection on specific spatial modes, such as transverse-electromagnetic modes, which are more practical to implement. These studies have shown that certain Gaussian measurement schemes can outperform direct imaging in some parameter regimes when the photon number per temporal mode is sufficiently large. 
However, the FIM obtained with Gaussian measurements still has the same scaling behavior as direct imaging, with any moderate improvement arising only from a different prefactor. These results indicate that superresolution cannot be achieved with the specific Gaussian measurements considered in their work.  But whether more sophisticated Gaussian measurements could achieve superresolution was an open question. 
Our results establish a general no-go theorem, definitively ruling out the possibility of any Gaussian measurement achieving superresolution in the imaging of  thermal sources. At the same time, our results place the performance of such schemes within the broader context of the fundamental limitations of Gaussian measurements in quantum imaging.  For comprehensive proofs of the no-go theorems for superresolution with Gaussian measurements—including both the two-point-source case and the general source case—as well as a review of earlier superresolution approaches and the no-go theorem for superresolution using interferometric imaging, see Sec.~\ref{Appendix:superresolution} of the Supplemental Material.

}

\textit{Conclusion and Discussion} - 
Our work establishes a new no-go result in quantum imaging, adding to the list of quantum tasks that cannot be achieved solely with Gaussian operations. We prove that the absence of non-Gaussianity fundamentally limits FIM in parameter estimation, leading to non-Gaussian measurements outperforming Gaussian ones by a factor of $\epsilon$ in estimation variance when imaging weak thermal sources with the mean photon number per temporal mode $\epsilon \rightarrow 0$. \yk{Interestingly,  the performance gap between Gaussian and non-Gaussian measurements identified by our results is as significant as the difference between nonlocal and local measurements in interferometric imaging, as discussed in Ref.~\cite{tsang2011quantum}, which has motivated the exploration of interferometric imaging based on quantum networks~\cite{gottesman2012longer,khabiboulline2019quantum,khabiboulline2019optical,huang2022imaging,marchese2023large,czupryniak2022quantum,czupryniak2023optimal,wang2023astronomical,purvis2024practical,huang2024limited}. And we also show that Gaussian measurements alone cannot achieve superresolution for any thermal sources with size $L\rightarrow0$.} Our no-go theorem unifies and clarifies previously scattered discussions on the performance of specific Gaussian and non-Gaussian measurements on Gaussian states in interferometric imaging~\cite{wang2023astronomical, purvis2024practical, huang2024limited} and single-lens imaging~\cite{nair2016far, lupo2016ultimate, wang2021superresolution, xie2024far, yang2017fisher, dong2020superresolution, yang2016far}, providing a powerful tool to understand the fundamental limitations of Gaussian measurements in quantum imaging. 

%We demonstrate the limitations of Gaussian measurement through illustrative examples in both interferometric and single-lens imaging.

%Our discussion focuses on imaging thermal sources, which are commonly encountered in nature and cover several active topics in quantum imaging. \yk{However, it is also interesting to consider scenarios where objects are actively illuminated, a widely used approach in the growing field of quantum microscopy \cite{taylor2016quantum,defienne2024advances,taylor2014subdiffraction,taylor2013biological,casacio2021quantum,cameron2024adaptive,ndagano2022quantum,tenne2019super,he2023quantum,classen2017superresolution,jin2018nanoparticles,huang2009super,picariello2025quantum,yue2025quantum,kudyshev2023machine,gatto2014beating,rust2006sub,betzig2006imaging,gustafsson2000surpassing,hell1994breaking}. It would be worthwhile to further investigate the imaging performance under the restriction of Gaussian illumination and Gaussian measurements.}

\yk{Our discussion focuses on imaging thermal sources, which are commonly encountered in nature and encompass several active topics in quantum imaging. However, active illumination represents an important complementary scenario, widely used in the rapidly developing field of quantum microscopy \cite{taylor2016quantum,defienne2024advances,taylor2014subdiffraction,taylor2013biological,casacio2021quantum,cameron2024adaptive,ndagano2022quantum,tenne2019super,he2023quantum,classen2017superresolution,jin2018nanoparticles,huang2009super,picariello2025quantum,yue2025quantum,kudyshev2023machine,gatto2014beating,rust2006sub,betzig2006imaging,gustafsson2000surpassing,hell1994breaking}. The key distinction is that active illumination can involve highly nonclassical states, such as squeezed light or two-photon entangled states. In contrast, our derivation assumes that the collected light is in a thermal state, which excludes such nonclassical illumination. Scenarios involving nonclassical active illumination therefore require separate analysis.  It would also be worthwhile to further investigate the imaging performance specifically under the restriction of Gaussian illumination and Gaussian measurements.}
Additionally, we can raise a broader question: what are the fundamental performance limits of using only Gaussian states and operations in various sensing tasks? For instance, it has been shown that homodyne detection can achieve Heisenberg scaling for parameter estimation in distributed sensing problems with Gaussian states as probes~\cite{zhuang2018distributed, wang2020continuous,oh2020optimal}. However, in some discussions, it has been observed that non-Gaussian measurements and states becomes necessary in the presence of loss for quantum sensing~\cite{oh2020optimal, gardner2024stochastic}. The limitations of Gaussian measurements in quantum sensing remain an open area of research.

\textit{Acknowledgements -} %S.Z. thanks Changhun Oh for helpful discussion. 
We thank Changhun Oh and Yujie Zhang for helpful discussion. Y.W. and S.Z. acknowledge funding provided by Perimeter Institute for Theoretical Physics, a research institute supported in part by the Government of Canada through the Department of Innovation, Science and Economic Development Canada and by the Province of Ontario through the Ministry of Colleges and Universities. Y.W. also acknowledges funding from the Canada First Research Excellence Fund.

%Our work adds a new no-go results to the list of quantum tasks which cannot be achieved only using Gaussian operations. Our no-theorem provides a unified and powerful tool to quantify the limitation of Gaussian measurement in quantum imaging, which can be applied on different topics for single-lens imaging or interferometric imaging. It clarify the condition where the lack of non-Gaussanity can lead to limitations in the Fisher information of estimating unknown parameters, the non-Gaussian measurement can outperform any Gaussian measurement by a factor of $\epsilon$ in the Fisher information in the case of imaging weak thermal sources with strength quantified by $\epsilon\rightarrow0$. The performance gap between Gaussian and non-Gaussian measurement can be huge similar to the difference between a nonlocal and local measurement as previously discussed in Ref.~\cite{tsang2011quantum}. We provide several examples in the case of interferometric imaging and single-lens imaging to demonstrate the performance limit of Gaussian measurement. As a side result, we also show Gaussian measurement solely cannot achieve superresolution.

\bibliography{arxiv}

\newpage

\appendix

\onecolumngrid

\section{Preliminary about the Gaussian state and measurement} \label{appendix:Gaussian_measurement}

In this section, we review  the formalism of Gaussian quantum information.
Any Gaussian measurement can be expressed in the form \cite{adesso2014continuous, weedbrook2012gaussian}
\begin{equation}\label{Pi_y}
\Pi_{\vec{y}}=\frac{1}{\pi^2}D_{\vec{y} }\Pi_0D^\dagger_{\vec{y}},
\end{equation}
where $\Pi_0$ is a Gaussian state with vanishing displacement and covariance matrix $V_\Pi$. The measurement outcome $\vec{y}$ is solely determined by the displacement of $\Pi_{\vec{y}}$. Our goal is to determine the probability distribution $P(\vec{y}|\vec{r}) = \tr(\Pi_{\vec{y}}\rho_{\vec{r}})$, where $\vec{r}$ represents the displacement of the state $\rho_{\vec{r}}$.
For two operators $A$ and $B$ with Wigner functions $W_A(\vec{q},\vec{p})$ and $W_B(\vec{q},\vec{p})$, their trace relation follows \cite{case2008wigner}:
$\tr[AB]\propto\int d\vec{q}d\vec{p}W_{A}(\vec{q},\vec{p})W_{B}(\vec{q},\vec{p})$. The Wigner functions for $\Pi_{\vec{y}}$ and $\rho_{\vec{r}}$ are given by
\begin{equation}
W_\Pi(q_1,p_1,q_2,p_2)=\frac{1}{\pi^2}\frac{\exp[-\frac{1}{2}(\vec{x}-\vec{y})^TV_\Pi^{-1}(\vec{x}-\vec{y})]}{(2\pi)^2\sqrt{\det V_{\Pi}}},
\end{equation}
\begin{equation}
W_\rho(q_1,p_1,q_2,p_2)=\frac{\exp[-\frac{1}{2}(\vec{x}-\vec{r})^TV_\rho^{-1}(\vec{x}-\vec{r})]}{(2\pi)^2\sqrt{\det V_\rho}},
\end{equation}
where $\vec{x}=[q_1,p_1,q_2,p_2]^T$, $\vec{y}=[y_1,y_2,y_3,y_4]^T$.
Since the integral is a standard Gaussian integral, we obtain
\begin{equation}\begin{aligned}
&P(\vec{y}|\vec{r})=\frac{1}{(2\pi)^2\sqrt{\det V}}\exp[-\frac{1}{2}(\vec{y}-\vec{r})^TV^{-1}(\vec{y}-\vec{r})],
\end{aligned}\end{equation}
where $V=V_\Pi+V_\rho$. 

For the Sudarshan-Glauber P representation of the form \cite{mandel1995optical}
\begin{equation}\begin{aligned}
&\rho_M = \int \frac{d^{2M}\vec{\gamma}}{\pi^M \det \Gamma}\,\exp\Bigl[-\vec{\gamma}^\dagger \Gamma^{-1}\vec{\gamma}\Bigr] \ket{\vec{\gamma}}\bra{\vec{\gamma}}, \\ 
\end{aligned}
\end{equation}
If we define the covariance matrix as $V_{ij}=\frac{1}{2}\langle\{\hat{z}_i,\hat{z}_j\}\rangle$, where $\{\hat{z}_i,\hat{z}_j\}=\hat{z}_i\hat{z}_j+\hat{z}_j\hat{z}_i$, $\langle \hat{O}\rangle=\tr(\rho_M\hat{O})$,  $\hat{\vec{z}}=[\hat{x}_1,\hat{x}_2,\cdots,\hat{x}_M,\hat{p}_1,\hat{p}_2,\cdots,\hat{p}_M]$, $\hat{a}_i=(\hat{x}_i+i\hat{p}_i)/\sqrt{2}$, the covariance matrix of $\rho_M$ is given by
\begin{equation}
V=\frac{1}{2}I_{2M}+\left[\begin{matrix}
\text{Re}\Gamma & -\text{Im}\Gamma \\
\text{Im}\Gamma & \text{Re}\Gamma 
\end{matrix}\right].
\end{equation}
Note that, for convenience, in the main text we order the quadrature operators as $\hat{\vec{z}} = [\hat{x}_1, \hat{p}_1, \hat{x}_2, \hat{p}_2, \dots, \hat{x}_M, \hat{p}_M]$, which results in a permutation of the elements of $V$.

%\section{Received weak thermal state using a single lens}\label{appendix:single lens state}

\section{No-go theorem for imaging weak thermal sources}\label{appendix:imaging_weak_source}

\subsection{\yk{Proof of Theorem \ref{thm:two_lens} for interferometric imaging of weak thermal sources}}\label{appendix:proof thm interferometric}

In the main text, we have derived the upper bound for $F_{|g||g|}$ in the case of interferometric imaging with two lenses. We now proceed to evaluate $F_{\theta\theta}$ and $F_{|g|\theta}$.
\begin{equation}\label{eq:pVpg}
\frac{\partial V}{\partial |g|}=V_{\partial |g|}=\frac{\epsilon}{2}\left[
\begin{matrix}
0 & 0 & \cos\theta & -\sin\theta\\    
0 & 0 & \sin\theta & \cos\theta\\
\cos\theta & \sin\theta & 0 & 0\\
-\sin\theta & \cos\theta & 0 & 0
\end{matrix}\right],
\end{equation}
\begin{equation}
\frac{\partial V}{\partial \theta}=V_{\partial\theta}=\frac{\epsilon |g|}{2}\left[
\begin{matrix}
0 & 0 & -\sin\theta & -\cos\theta\\    
0 & 0 & \cos\theta & -\sin\theta\\
-\sin\theta & \cos\theta & 0 & 0\\
-\cos\theta & -\sin\theta & 0 & 0
\end{matrix}\right].
\end{equation}
Define
\begin{equation}
U_2=\frac{1}{\sqrt{2}}\left[
\begin{matrix}
\cos\theta & \sin\theta & 0 & 1\\    
\sin\theta & -\cos\theta & 1 & 0\\
-\cos\theta & -\sin\theta & 0 & 1\\
-\sin\theta & \cos\theta & 1 & 0
\end{matrix}\right].
\end{equation}
We can then define
\begin{equation}
\begin{aligned}
&\Sigma'^N=(I_N\otimes U_2)V^N(I_N\otimes U_2)^\dagger,\quad \Sigma'^N_\Pi=(I_N\otimes U_2)V^N_\Pi(I_N\otimes U_2)^\dagger,\\
&\Sigma_{\partial\theta}=U_2\frac{\partial V}{\partial \theta} U_2^\dagger=\frac{\epsilon|g|}{2}\text{diag}[-1,-1,1,1],\quad\Sigma^N_{\partial\theta}=I_N\otimes\Sigma_{\partial\theta}.
\end{aligned}
\end{equation}
We can then follow a similar proof approach
\begin{equation}\label{eq:F_thetatheta}
\begin{aligned}
F_{\theta\theta}&=\frac{1}{2}\tr((\Sigma'^N)^{-1}\Sigma^N_{\partial\theta}(\Sigma'^N)^{-1}\Sigma^N_{\partial\theta})\\
&\leq \frac{\epsilon^2|g|^2}{8}\tr((\Sigma'^N)^{-2})\\
&\leq  \frac{\epsilon^2|g|^2}{8}\tr((\Sigma_\rho^N)^{-2})=2\epsilon^2 |g|^2N,
\end{aligned}\end{equation}
where, in the first inequality, we take the absolute values of the eigenvalues of $\Sigma^N_{\partial\theta}$ because, using the spectral decomposition $\Sigma^N_{\partial\theta} = \sum_i \lambda_i \ket{v_i} \bra{v_i}$, we obtain $F_{\theta\theta}=\frac{1}{2}\sum_{i,j}\lambda_i\lambda_j|\bra{v_i}(\Sigma'^N)^{-1}\ket{v_j}|^2\leq \frac{1}{2}\sum_{i,j}|\lambda_i\lambda_j||\bra{v_i}(\Sigma'^N)^{-1}\ket{v_j}|^2$. The second inequality follows from the fact that $\Sigma'_\Pi$ is positive semidefinite.

For the off-diagonal elements of the FIM $F_{|g|\theta}$, since the FIM $F$ is positive semidefinite, it follows from the properties of positive semidefinite matrices that we have
\begin{equation}\begin{aligned}
&F_{|g|\theta}\leq \sqrt{F_{|g||g|}F_{\theta\theta}}=2\epsilon^2|g|N.
\end{aligned}\end{equation}
%And the Fisher information matrix
%\begin{equation}
%F=\left[
%\begin{matrix}
%F_{|g||g|} & F_{|g|\theta}\\
%F_{\theta|g|} & F_{\theta\theta}
%\end{matrix}\right]\leq (F_{|g||g|}+F_{\theta\theta})I\leq O(N\epsilon^2)I
%\end{equation}

We now aim to establish an upper bound on the FIM using matrix inequalities, where $A \geq B$ denotes that $A - B$ is positive semidefinite.
As positive semidefinite symmetric matrix, FIM $F$ has its eigenvalues bounded by
\begin{equation}\begin{aligned}\label{eq:A_inequality}
&\lambda_{\text{max}}(F)\leq\tr(F)  .
\end{aligned}
\end{equation}
So, we can easily find the FIM is bounded by
\begin{equation}
F\leq 2N\epsilon^2(1+|g|^2) I_2,
\end{equation}
in the sense of matrix inequality.

Having established the case of interferometric imaging with two lenses, we now extend our analysis to the more general scenario with similar observations. 
In the case of interferometric imaging with $M$ lense, the received state is
\begin{equation}\begin{aligned}
\label{rho_M}
&\rho_M = \int \frac{d^{2M}\vec{\gamma}}{\pi^M \det \Gamma}\,\exp\Bigl[-\vec{\gamma}^\dagger \Gamma^{-1}\vec{\gamma}\Bigr] \ket{\vec{\gamma}}\bra{\vec{\gamma}}, \quad d^{2M}\vec{\gamma} = \prod_{j=1}^{M} d^2\gamma_j, \quad \vec{\gamma}=[\gamma_1,\gamma_2,\cdots,\gamma_M],\\
&\ket{\vec{\gamma}}=\prod_{j=1}^M\exp(\gamma_j \hat{a}_j^\dagger-\gamma_j^*\hat{a}_j)\ket{0},\quad \Gamma_{ij} = \frac{\epsilon}{2} \times
\begin{cases}
1, & \text{if } i=j, \\
g_{ij}, & \text{if } i\neq j,
\end{cases}\\
&V_\rho=\frac{1}{2}I_{2M}+G,\quad G=\left[
\begin{matrix}
\text{Re}\Gamma & -\text{Im} \Gamma\\
\text{Im} \Gamma & \text{Re}\Gamma
\end{matrix}\right],
\end{aligned}
\end{equation}
where $g_{ij}=|g_{ij}| e^{i\theta_{ij}}$, $g_{ij}=g_{ji}^*$. 
We adopt the convention that as the number of telescopes $M$ increases, the total mean photon number per temporal mode is given by $\tr\Gamma=\epsilon M/2$, which scales linearly with $M$ for convenience. This choice is justified, as increasing the number of telescopes effectively expands the light-collecting area, leading to a proportional increase in the number of collected stellar photons.

Note that when computing $\frac{\partial V}{\partial |g_{ij}|}$ for each $i,j$, the nonvanishing elements are exactly the same as those in Eq.~\ref{eq:pVpg}. A similar result holds for $\frac{\partial V}{\partial |\theta_{ij}|}$.
\begin{equation}
\frac{\partial V}{\partial |g_{ij}|}=V_{\partial |g_{ij}|}=\left[
\begin{matrix}
V_{\partial|g|} & 0\\
0 & 0
\end{matrix}\right],
\end{equation}
where $V_{\partial|g|}$ is given by Eq.~\ref{eq:pVpg}, with $\theta$  in Eq.~\ref{eq:pVpg} replaced by $\theta_{ij}$, and the order of elements has been adjusted for convenience. We can then define
\begin{equation}
U_1'=\left[
\begin{matrix}
U_1 & 0\\
0 & I
\end{matrix}\right].
\end{equation}
And similarly, we define $\Sigma_{\partial |g_{ij}|} = U_1' V_{\partial|g_{ij}|} U_1'^\dagger$. Considering $N$ copies of the state $\rho_M^{\otimes N}$, we can define $\Sigma^N$, $\Sigma^N_{\partial|g_{ij}|}$, and $\Sigma_\Pi^N$ by tensoring with $I_N$. We then have
\begin{equation}
\begin{aligned}
F_{|g_{ij}||g_{ij}|}&=\frac{1}{2}\tr((\Sigma^N)^{-1}\Sigma^N_{\partial|g_{ij}|}(\Sigma^N)^{-1}\Sigma^N_{\partial|g_{ij}|})\\
&\leq \frac{\epsilon^2}{8}\tr[(P(\Sigma^N)^{-1}P)^2]\\
&\leq \frac{\epsilon^2}{8}\tr[(P(\Sigma_\rho^N)^{-1}P)^2],\\
&P=I_N\otimes\left[
\begin{matrix}
I_4 & 0\\
0 & 0
\end{matrix}\right],
\end{aligned}\end{equation}
where, in the first inequality, we take the absolute values of the eigenvalues of $\Sigma^N_{\partial|g_{ij}|}$, similar to Eq.~\ref{eq:F_thetatheta}, $P$ is the projector onto the support of $\Sigma^N_{\partial|g_{ij}|}$. In the second inequality, we use the fact that $(\Sigma^N)^{-1} \leq (\Sigma^N_{\rho})^{-1}$. Since any diagonal block of a positive semidefinite matrix is also positive semidefinite, it follows that $P(\Sigma^N)^{-1} P \leq P(\Sigma^N_{\rho})^{-1} P$.
Since $G$ is positive semidefinite, we find that all of the eigenvalues of $(\Sigma^N_{\rho})^{-1}$ are bounded by $\lambda_i((\Sigma^N_{\rho})^{-1}) \leq 2 $.   Defining $A = P(\Sigma_\rho^N)^{-1} P$, we note that $\text{rank}(A) \leq 4N$ and that $\tr(A^2) \leq 4N \lambda_{\text{max}}(A^2)$. We now proceed to determine the largest eigenvalue of $A$.
\begin{equation}\begin{aligned}
&\lambda_{\text{max}}(A)=x^TAx=(Px)^T(\Sigma^N_{\rho})^{-1}(Px)\leq \lambda_{\text{max}}((\Sigma^N_{\rho})^{-1})||Px||^2\\
&\leq \lambda_{\text{max}}((\Sigma^N_{\rho})^{-1})||x||^2=\lambda_{\text{max}}((\Sigma^N_{\rho})^{-1})\leq 2,\\
\end{aligned}
\end{equation}
We have thus proved 
\begin{equation}
F_{|g_{ij}||g_{ij}|}\leq 2N\epsilon^2.
\end{equation}
We can also similarly prove
\begin{equation}
F_{\theta_{ij}\theta_{ij}}\leq 2N|g_{ij}|^2\epsilon^2.
\end{equation}
As FIM $F$ is a positive semidefinte matrix, we find that all off-diagonal elements of the FIM $F$ are also upper bounded by $NO(\epsilon^2)$. Thus, we have proven that the FIM scales as $NO(\epsilon^2)$ in the more general case of $M$ lenses in interferometric imaging.

As a multiparameter estimation problem with $M(M-1)$ unknown parameters in the case of $M$ telescopes, we can also bound the FIM using a matrix inequality based on Eq.~\ref{eq:A_inequality},
\begin{equation}
F\leq \sum_{i>j}2N\epsilon^2(1+|g_{ij}|^2)I_{M(M-1)}\leq 2M(M-1)N\epsilon^2I_{M(M-1)},
\end{equation}
where $i,j=1,2,\cdots, M$. Since the number of lenses $M$ in interferometric imaging is always finite, we conclude that the FIM scales as $NO(\epsilon^2)$ in the sense of matrix inequality.

\yk{As illustrated in Fig. \ref{fig:comparison}, Ref. \cite{tsang2011quantum} shows a scaling gap along the nonlocal and local measurement direction. Theorem~\ref{thm:two_lens} generalizes this result by identifying an analogous scaling gap between Gaussian and non-Gaussian measurements. The FIM reaches $N O(\epsilon)$, as indicated by the red shaded region, only when the measurement is both nonlocal and non-Gaussian.}

\begin{figure}[!tb]
\begin{center}
\includegraphics[width=0.5\columnwidth]{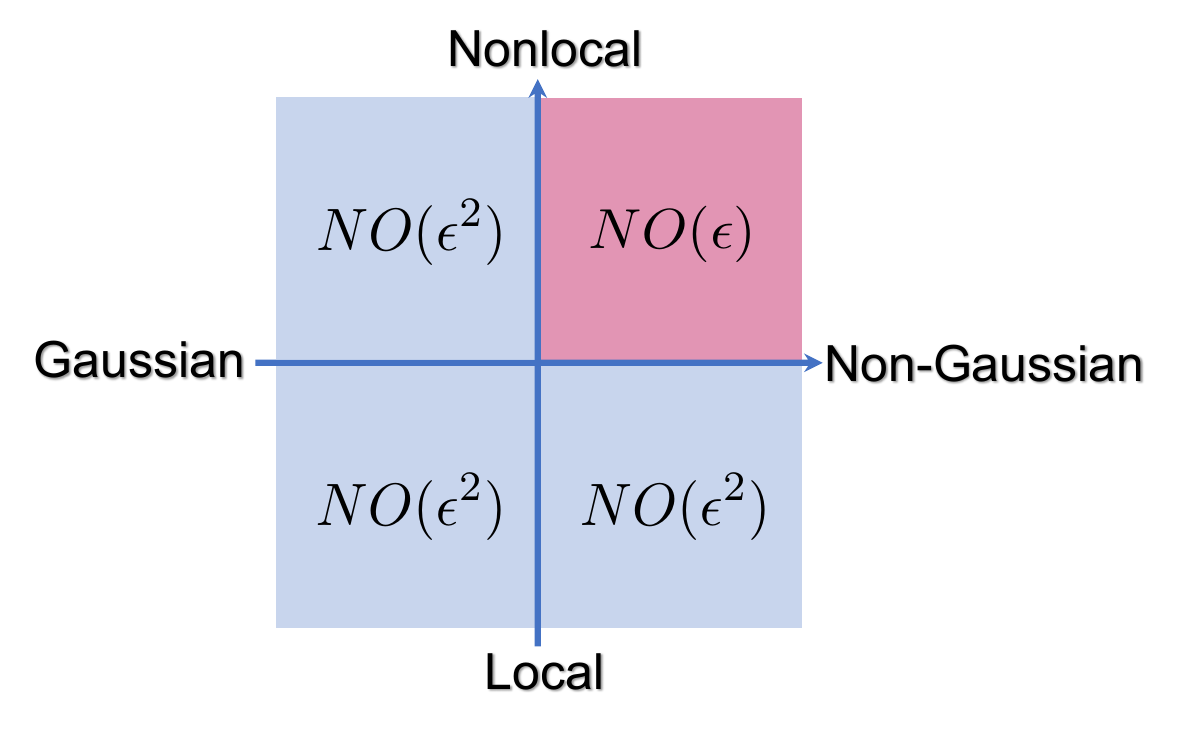}
\caption{The Fisher information matrix scaling using different measurement for interferometric imaging.}
\label{fig:comparison}
\end{center}
\end{figure}

\subsection{Proof of no-go theorem for imaging weak thermal sources using a single lens}\label{appendix:proof thm single lens}

\yk{
In interferometric imaging, we discuss the performance gap between local and nonlocal measurements. In contrast, single-lens imaging does not involve a comparable distinction between local and nonlocal measurements. We first provide a more detailed explanation of why this distinction between local and nonlocal measurement is meaningful and necessary for interferometric imaging, but not for single-lens imaging.

In a non-rigorous, conceptual picture, single-lens imaging can be regarded as a special case of interferometric imaging, where the lens aperture acts like a continuum of infinitesimally small, closely spaced sub-apertures. Each point on the aperture transmits part of the incoming wavefront, and because the field is coherent across the aperture, these contributions interfere at the detection plane. The lens imposes a spatially varying phase shift so that light from a given object point is redirected to a corresponding point in the image, forming the image through interference. Recall that, in interferometric imaging, direct interference between light from different apertures is a canonical example of  nonlocal measurements. In this sense, the natural interference of light from different sub-apertures at the detection plane inherently realizes a nonlocal measurement in the interferometric framework. But instead of separately detecting each sub-aperture and then extracting coherence information, the optical system  produces the image directly on the detection plane. Consequently, if single-lens imaging is viewed through the interferometric framework, it always corresponds to the ``nonlocal measurement'' case. However, because this nonlocality is built into the optical hardware and there is no meaningful choice between local and nonlocal detection, the distinction is usually not meaningful for single-lens imaging.

We emphasize that, in this conceptual picture, the longest baseline in single-lens imaging is simply the lens diameter. Since the resolution in interferometric imaging is determined by the baseline between different apertures, this picture likewise predicts that the resolution of single-lens imaging is set by its aperture diameter. The advantage of interferometric imaging lies in its ability to exploit much longer baselines between spatially separated apertures, making the role of nonlocal measurements nontrivial—unlike in the single-lens case, where such measurements are inherently realized by the optical system.

We now proceed to prove the no-go theorem for single-lens imaging of weak thermal sources using Gaussian measurements.
We first justify the form of thermal states on the detection plane in a single-lens imaging system.}
Assuming that the mutual coherence matrix at the source plane is given by $\Gamma^{(o)}$ and that $S$ represents the field scattering matrix, which characterizes the response of the imaging system to each point on the source, the mutual coherence matrix on the image plane is given by \cite{mandel1995optical}
\begin{equation}
\Gamma=S\Gamma^{(o)}S^\dagger.
\end{equation}
In the case of incoherent source, we have
\begin{equation}
\Gamma^{(o)}_{ij}=\delta_{ij}n_i,\quad \Gamma_{jk}=\sum_i n_i S_{ji}S^*_{ki},
\end{equation}
where $n_i$ represents the intensity of the $i$th point on the source. Defining the quantum efficiency for detecting each point on the source as $\eta_i = \sum_j |S_{ji}|^2$, we introduce the normalized field scatter matrix as $\psi_{ji} = S_{ji} / \sqrt{\eta_i}$.
If we define $\epsilon = \sum_i \eta_i n_i$, we obtain
\begin{equation}
\Gamma_{jk}=\epsilon\sum_i J_i \psi_{ji}\psi^*_{ki},
\end{equation}
where $J_i=\eta_in_i/(\sum_i \eta_in_i)$.

Its covariance matrix is given by
\begin{equation}
V_\rho=\frac{1}{2}I_{2M}+G,\quad G=\left[
\begin{matrix}
\text{Re}\Gamma & -\text{Im} \Gamma\\
\text{Im} \Gamma & \text{Re}\Gamma
\end{matrix}\right],\quad V_{\partial\theta}=\frac{\partial V}{\partial\theta},\quad G_{\partial\theta}=\frac{\partial G}{\partial\theta},
\end{equation}

For the non-Gaussian measurement, which projects onto $\ket{j}$, the probability of obtaining the outcome $\ket{j}$ is given by
\begin{equation}\begin{aligned}
&P(j)=\bra{j}\rho_W\ket{j}=\int \frac{d^{2W}\vec{\gamma}}{\pi^W \det \Gamma}\bra{0}a_j \ket{\vec{\gamma}}\bra{\vec{\gamma}}a_j^\dagger\ket{0}\\
&=\int \frac{d^{2W}\vec{\gamma}}{\pi^W \det \Gamma}\,\exp\Bigl[-\vec{\gamma}^\dagger \Gamma^{-1}\vec{\gamma}\Bigr] \exp\Bigl[-\sum_k|\gamma_k|^2\Bigr]|\gamma_j|^2\\
&=\frac{\det A}{\det \Gamma}A_{jj}=\Gamma_{jj}+O(\epsilon^2),\quad A^{-1}=\Gamma^{-1}+I,\\
&A=(I+\Gamma)^{-1}\Gamma, \quad (I+\Gamma)^{-1}=I-\Gamma+O(\epsilon^2),
\end{aligned}
\end{equation}
where we expand $(I+\Gamma)^{-1}$ in the limit $\|\Gamma\|\rightarrow0$ \cite{horn2012matrix}. Since the detection of vacuum states provides no information about $\theta$ and the probability of detecting more than one photon is of order $O(\epsilon^2)$, we obtain the FI for the non-Gaussian measurement in single-lens imaging, as presented in the main text.

For any Gaussian measurement, 
\begin{equation}\begin{aligned}\label{eq:F_theta_single_lens}
&F_{\theta\theta}\leq \frac{1}{2}\tr((V^N)^{-1}V^N_{\partial\theta}(V^N)^{-1}V^N_{\partial\theta})\\
&\leq \lambda_{\text{max}}((V^N_{\partial\theta})^2)\frac{1}{2}\tr((P(V^N)^{-1}P)^2) \\
&\leq\lambda_{\text{max}}((V^N_{\partial\theta})^2)\frac{1}{2}\tr((P(V^N_\rho)^{-1}P)^2), \\
\end{aligned}
\end{equation}
in the first inequality, we use the spectral decomposition $V^N_{\partial\theta} = \sum_i \lambda_i \ket{v_i} \bra{v_i}$ (note $V^N_{\partial\theta}$ is a real symmetric matrix), we obtain $F_{\theta\theta}=\frac{1}{2}\sum_{\lambda_{i}\lambda_j\neq0}\lambda_i\lambda_j|\bra{v_i}(V^N)^{-1}\ket{v_j}|^2\leq \frac{1}{2}|\lambda_{\text{max}}|^2\sum_{\lambda_{i}\lambda_j\neq0}|\bra{v_i}(V^N)^{-1}\ket{v_j}|^2$. $P$ is the projection onto the support of $V^N_{\partial\theta}$. The second inequality follows from the fact that $V^N \geq V_\rho^N$ and that a diagonal block of a positive semidefinite matrix remains positive semidefinite.
Since $G$ is positive semidefinite, all of the eigenvalues of $(V^N_{\rho})^{-1}$ satisfy $\lambda_i((V^N_{\rho})^{-1}) \leq 2$. Define $A = P (V_\rho^N)^{-1} P$. Clearly, $\text{rank}(A) \leq \text{rank}(P)$, and by definition, $\text{rank}(P) = N \text{rank}( V_{\partial\theta})$. Moreover, we have the bound $\tr(A^2) \leq \text{rank}(V_{\partial\theta}) N \lambda_{\text{max}}(A^2)$.
We now determine the largest eigenvalue of $A$,
\begin{equation}\begin{aligned}\label{eq:PVP}
&\lambda_{\text{max}}(A)=x^TAx=(Px)^T(V^N_{\rho})^{-1}(Px)\leq \lambda_{\text{max}}((V^N_{\rho})^{-1})||Px||^2\\
&\leq \lambda_{\text{max}}((V^N_{\rho})^{-1})||x||^2=\lambda_{\text{max}}((V^N_{\rho})^{-1})\leq 2,\\
\end{aligned}
\end{equation}
\begin{equation}\begin{aligned}
&F_{\theta\theta}\leq\lambda_{\text{max}}((V^N_{\partial\theta})^2) \text{rank}( V_{\partial\theta})2N\\
&=\max_i \left|\lambda_i(G_{\partial\theta})\right|^2 \text{rank}(G_{\partial\theta})2N=NO(\epsilon^2).
\end{aligned}
\end{equation}

We can similarly extend this discussion to the case of multiparameter estimation, where $\vec{\theta} = [\theta_1, \theta_2, \dots, \theta_Q]$. The elements of the FIM can be bounded as follows,
\begin{equation}\begin{aligned}
&F_{\theta_i\theta_i}\leq \max_j \left|\lambda_j(G_{\partial\theta_i})\right|^2 \text{rank}(G_{\partial\theta_i})2N=NO(\epsilon^2).
\end{aligned}
\end{equation}
Similarly, we have $F_{\theta_i\theta_j} \leq \sqrt{F_{\theta_i\theta_i} F_{\theta_j\theta_j}}$. Using Eq.~\ref{eq:A_inequality}, we can further bound the FIM in the sense of a matrix inequality
\begin{equation}
F\leq \sum_{i=1}^Q \max_j \left|\lambda_j(G_{\partial\theta_i})\right|^2 \text{rank}(G_{\partial\theta_i})2N \,I_Q.
\end{equation}
If the number of unknown parameters to be estimated remains finite and independent of $\epsilon$, this factor does not affect the claimed $NO(\epsilon^2)$ scaling.

\subsection{Examples}

\subsubsection{Example of interferometric imaging using non-Gaussian measurement}\label{appendix:non-Gaussian}

In this subsection, we analyze a possible non-Gaussian measurement in the context of interferometric imaging with two lenses, which receive weak thermal states as described in the main text
\begin{equation}\begin{aligned}\label{rho_thermal}
&\rho=\int\frac{d^2\alpha d^2\beta}{\pi^2\det\Gamma}\exp(-\vec{\gamma}^\dagger \Gamma^{-1}\vec{\gamma})\ket{\vec{\gamma}}\bra{\vec{\gamma}},\\
&\vec{\gamma}=[\alpha,\beta]^T,\quad \Gamma=  \frac{\epsilon}{2}\left[
\begin{matrix}
1 & g\\
g^* & 1
\end{matrix}\right],\\
&\ket{\vec{\gamma}}=\exp(\alpha \hat{a}^\dagger-\alpha^*\hat{a})\exp(\beta \hat{b}^\dagger-\beta^*\hat{b})\ket{0}.
\end{aligned}\end{equation}
If we project the state onto the basis
\begin{equation}
\ket{\pm}=(\ket{01}+e^{i\delta}\ket{10})/\sqrt{2},
\end{equation}
where $\ket{01}$ and $\ket{10}$ represent single-photon states in the two spatial modes, respectively, and $\delta$ is a phase delay that can be chosen in the measurement. Note that this measurement can be implemented by first combining the light received by the two lenses on a beam splitter, followed by single-photon detection at the two output ports. This constitutes a nonlocal measurement.
The probability distribution can then be determined as
\begin{equation}
P(\pm)=\frac{\epsilon}{2}(1\pm|g|\cos(\theta+\delta)),
\end{equation}
which gives the FIM for estimating the amplitude $|g|$ and phase $\theta$ of the coherence function $g = |g| e^{i\theta}$ on $N$ copies of the state as
\begin{equation}
F=\frac{N\epsilon}{1-|g|^2\cos^2(\theta+\delta)}\left[\begin{matrix}
\cos^2(\theta+\delta) & -|g|\sin(\theta+\delta)\cos(\theta+\delta)\\
-|g|\sin(\theta+\delta)\cos(\theta+\delta) & |g|^2\sin^2(\theta+\delta)
\end{matrix}\right].
\end{equation}
Note that this matrix is not full rank, but this issue can be resolved by using two different phase delays, $\delta$. Importantly, the FIM in this case scales as $N\Theta(\epsilon)$.

\subsubsection{Example of interferometric imaging with Gaussian measurement}\label{appendix:examples}

We continue to examine measurements on the two-mode weak thermal state in Eq.~\ref{rho_thermal} for interferometric imaging with two lenses. To implement a nonlocal Gaussian measurement, we first combine the two modes on a beam splitter. At the two output ports, we consider both homodyne and heterodyne detection.

In the case of homodyne detection, we consider four different scenarios where the two output ports, labeled as 1 and 2,  which measures the observables $\hat{x}_1$ or  $\hat{p}_1$ and $\hat{x}_2$, or $\hat{p}_2$, respectively.

When we detect, $\hat{x}_1,\hat{x}_2$ or $\hat{p}_1, \hat{p}_2$,  we can describe the measurement on the states at the two output ports of the beam splitter with the covariance matrix of the POVM
\begin{equation}\begin{aligned}
V_\Pi=\text{diag}[0,\infty,0,\infty], \quad\text{for}\quad \hat{x}_1,\hat{x}_2    \\
V_\Pi=\text{diag}[\infty,0,\infty,0], \quad\text{for}\quad \hat{p}_1,\hat{p}_2 
\end{aligned}
\end{equation}
the FIM is given by
\begin{equation}\begin{aligned}
F=&
\frac{N
  \epsilon^2  \left( 2 + \epsilon \left( 4 + \epsilon \left( 2 + |g|^2 \right) \right) + \epsilon^2 |g|^2 \cos(2\theta) \right)
}{
  \left((1+\epsilon)^2 - \epsilon^2 |g|^2 \cos^2\theta\right)^2
}\left[\begin{matrix}
\cos^2\theta & -\sin 2\theta\\
-\sin 2\theta & \sin^2\theta
\end{matrix}\right],    
\end{aligned}
\end{equation}
where the first and second parameters of $F$ is $|g|$ and $\theta$.

When we detect, $\hat{x}_1,\hat{p}_2$ or $\hat{p}_1, \hat{x}_2$, we can describe the measurement on the states at the two output ports of the beam splitter with the covariance matrix of the POVM
\begin{equation}\begin{aligned}
V_\Pi=\text{diag}[0,\infty,\infty,0], \quad\text{for}\quad \hat{x}_1,\hat{p}_2    \\
V_\Pi=\text{diag}[\infty,0,0,\infty], \quad\text{for}\quad \hat{p}_1,\hat{x}_2 
\end{aligned}
\end{equation}
\begin{equation}
F_{|g||g|}=N\epsilon^2 \left(\frac{1}{\left(1+\epsilon-\epsilon|g|\right)^2} + \frac{1}{\left(1+\epsilon+\epsilon|g|\right)^2}\right),
\end{equation}
\begin{equation}
F_{\theta\theta}=\frac{2 N\epsilon^2 |g|^2}{1+\epsilon\left(2+\epsilon-\epsilon|g|^2\right)},
\end{equation}
\begin{equation}
F_{|g|\theta}=0.
\end{equation}
When they do heterodyne detection, which projects onto the basis $\{\frac{1}{\pi^2}\ket{\alpha\beta}\bra{\alpha\beta}\}$, we can describe the measurement on the states at the two output ports of the beam splitter with the covariance matrix of the POVM 
\begin{equation}\begin{aligned}
V_\Pi=\frac{1}{2}I_4,
\end{aligned}
\end{equation}
the FIM is given by
\begin{equation}
F_{|g||g|}=2N\epsilon^2 \left(\frac{1}{\left(2+\epsilon-\epsilon|g|\right)^2} + \frac{1}{\left(2+\epsilon+\epsilon|g|\right)^2}\right),
\end{equation}
\begin{equation}
F_{\theta\theta}=\frac{4N\epsilon^2 |g|^2}{4+\epsilon\left(4+\epsilon-\epsilon|g|^2\right)},
\end{equation}
\begin{equation}
F_{|g|\theta}=0.
\end{equation}
Thus, for all the Gaussian measurements considered above, the FIM scales as $NO(\epsilon^2)$, consistent with the general proof.

\section{No-go theorem for superresolution using Gaussian measurement}\label{Appendix:superresolution}

\yk{
\subsection{Review of  superresolution}
\label{SI:preliminary_superresolution}

Superresolution exploits the fact that a carefully designed measurement can yield a much larger FIM than direct imaging when the source size is well below the Rayleigh limit of the imaging system. Ref.~\cite{tsang2016quantum} showed that, for estimating the separation between two point sources, the FI of direct imaging vanishes as the separation approaches zero, reflecting Rayleigh's limit. In contrast, the FI for a spatial-mode demultiplexing (SPADE) measurement in the Hermite–Gaussian basis remains constant, demonstrating that an appropriately chosen measurement strategy can circumvent Rayleigh's limit.

Later on, this discussion is extended to consider a general source whose size $L\rightarrow 0$ \cite{zhou2019modern,tsang2017subdiffraction,tsang2019quantum}. We now review the discussion in Ref. \cite{zhou2019modern} using our notation in more detail for later use. A general incoherent source can be modeled as  
\begin{equation}
\rho=\int dy dx_1 dx_2 \zeta(y)\psi(y-x_1)\ket{x_1}\bra{x_2}\psi(y-x_2),
\end{equation}
where $\ket{x}=a^\dagger_x\ket{0}$ is the single photon state at position $x$, $\psi(x)$ is the point spread function (PSF).  Assume the normalized source intensity $\zeta(y)$ is confined within the interval $[-L/2+y_0, L/2+y_0]$, we expand the PSF 
\begin{equation}
\psi(x_1-y)=\sum_{n=0}^\infty \left(\frac{\partial^n \psi(x_1-y)}{\partial y^n}\bigg|_{y=y_0}\frac{L^n}{n!}\right)\left(\frac{y-y_0}{L}\right)^n=:\sum_{n=0}^\infty \psi^{(n)}(x_1)\left(\frac{y-y_0}{L}\right)^n.
\end{equation}
We can then get
\begin{equation}\begin{aligned}
\rho&=\sum_{m,n=0}^\infty\left(\int du \zeta(y)\left(\frac{y-y_0}{L}\right)^{m+n}\right)\left(\int dx_1\psi^{(m)}(x_1)\ket{x_1}\right)\left(\int dx_2\psi^{(m)}(x_2)\bra{x_2}\right)\\
&=\sum_{m,n=0}^\infty t_{m+n}\ket{\psi^{(m)}}\bra{\psi^{(n)}}=\sum_{k=0}^\infty t_k\rho^{(k)},
\end{aligned}\end{equation}
where $\rho^{(k)}:=\sum_{m+n=k}\ket{\psi^{(m)}}\bra{\psi^{(n)}}$, the $k$th moment is defined as $t_k:=\int dy \zeta(y)\left(\frac{y-y_0}{L}\right)^k$.
Our goal is to design a measurement that suppresses the dominant noise arising from the lower-order terms of $L$ when estimating higher-order moments. To accomplish this, we build an orthonormal set of measurement states ${\ket{b_l}}_l$ using the Gram–Schmidt process,  
\begin{equation}
a_{ml}=\bra{\psi^{(m)}}\ket{b_l}\left\{
\begin{array}{cc}
=0    &  m\leq l-1\\
\neq 0     &  m\geq l
\end{array}\right.
\end{equation}
The measurement is then constructed as 
\begin{equation}\label{ref5_POVM}
\left\{\frac{1}{2}\ket{\phi_{i,\pm}}\bra{\phi_{i,\pm}},\frac{1}{2}\ket{b_0}\bra{b_0}\right\}, 
\end{equation}
\begin{equation}
\ket{\phi_{i,\pm}}=(\ket{b_i}\pm\ket{b_{i+1}})/\sqrt{2},\quad i=0,1,2,\cdots
\end{equation}  
The performance of the above measurement is explicitly quantified by the FIM.  
\begin{equation}\label{Fij_suboptimal}
F_{t_it_j}=\Theta(L^{i+j-2\lfloor \min\{i,j\}/2 \rfloor}).
\end{equation}  
In particular, the diagonal elements $F_{t_it_i}$ scale as $\Theta(L^{0}), \Theta(L^{2}), \Theta(L^{2}), \Theta(L^{4}), \Theta(L^{4}), \cdots$ for $i = 0, 1, 2, 3, 4, \cdots$.

In direct imaging, the measurement corresponds to projecting onto $\ket{x}\bra{x}$ to obtain the intensity distribution at each position $x$ on the detection plane. The FIM for estimating $t_n$ can then be computed as
\begin{equation}
F_{t_it_j}=\Theta(L^{i+j}).
\end{equation}
This value is significantly smaller than that obtained with the measurement in Eq.~\ref{Fij_suboptimal}. In particular, the diagonal elements scale as $F_{t_i t_i} = \Theta(L^{0}), \Theta(L^{2}), \Theta(L^{4}), \Theta(L^{6}), \Theta(L^{8}), \ldots$ for $i = 0, 1, 2, 3, 4, \ldots$, highlighting that a properly designed measurement strategy can substantially enhance sensitivity.

}

\subsection{Proof of the no-go theorem for the superresolution of two point sources
}\label{appendix:proof thm superresolution}
We first derive the form of the state received by a single lens when imaging two weak thermal point sources of equal strength located at positions $\pm L/2$. On the source plane, before entering the imaging system, the state is given by
\begin{equation}\begin{aligned}
&\rho=\int\frac{d^2\alpha d^2\beta}{\pi^2\det\Gamma_0}\exp(-[\alpha^*,\beta^*] \Gamma_0^{-1}[\alpha,\beta]^T)\ket{\alpha}\bra{\alpha}\otimes\ket{\beta}\bra{\beta},\\
&\Gamma_0=  \frac{\epsilon}{2}\left[
\begin{matrix}
1 & 0\\
0 & 1
\end{matrix}\right],\\
&\ket{\alpha}\otimes\ket{\beta}=\exp(\alpha \hat{a}^\dagger-\alpha^*\hat{a})\exp(\beta \hat{b}^\dagger-\beta^*\hat{b})\ket{0}.
\end{aligned}\end{equation}
We derive the state received on the imaging plane by evolving the mode operators through the imaging system.
\begin{equation}
a\rightarrow \int dx \psi(x-L/2)c_x,\quad b\rightarrow \int dx \psi(x+L/2)c_x,
\end{equation}
\yk{where $c_x$ is the annihilation operator at position $x$ on the detection plane, $\psi(x)$ is the PSF and can have any general shape.} The state received at the detection plane is given by
\begin{equation}\begin{aligned}
&\rho=\int\frac{\prod_{i=1}^Wd^2\gamma_{x_i}}{\det(\pi\Gamma)}\exp(-\vec{\gamma}^\dagger \Gamma^{-1}\vec{\gamma})\ket{\vec{\gamma}}\bra{\vec{\gamma}},\\
&\vec{\gamma}=[\gamma_{x_1},\gamma_{x_2},\cdots,\gamma_{x_W}]^T,\quad \Gamma=R\Gamma_0R^\dagger=\frac{\epsilon}{2}(\psi_0\psi_0^\dagger+\psi_1\psi_1^\dagger), \\
&R=\left[
\begin{matrix}
\psi(x_1-L/2) & \psi(x_1+L/2) \\
\psi(x_2-L/2) & \psi(x_2+L/2) \\
\vdots & \vdots\\
\psi(x_W-L/2) & \psi(x_W+L/2) \\
\end{matrix}\right]=[\psi_0,\psi_1],\\
&\ket{\vec{\gamma}}=\exp(\sum_i\gamma_{x_i}c_{x_i}^\dagger-\gamma_{x_i}^*c_{x_i})\ket{0},
\end{aligned}\end{equation}
Since $\Gamma$ is real for the chosen PSF, the covariance matrix of the state is given by
\begin{equation}
V_\rho=\frac{1}{2}I_{2W}+G,\quad G=I_2\otimes \Gamma.
\end{equation}
The FIM element for estimating the separation $L$ from $\rho^{\otimes N}$ is given by (as before, we use $V^N, G^N, \dots$ to denote the tensor product with $I_N$ when considering $N$ copies of the state)
\begin{equation}
F_{LL}=\frac{1}{2}\tr((V^N)^{-1}\frac{\partial V^N}{\partial L}(V^N)\frac{\partial V^N}{\partial L})=\frac{1}{2}\tr((V^N)^{-1}\frac{\partial G^N}{\partial L}(V^N)^{-1}\frac{\partial G^N}{\partial L}),
\end{equation}
where $V^N=V^N_\rho+V^N_\Pi$, $V^N_\Pi$ is covariance matrix of the Gaussian measurement. 
Note that
\begin{equation}\begin{aligned}
&\psi_0=e^{(0)}-\frac{L}{2}e^{(1)}+\frac{L^2}{8}e^{(2)}+o(L^2),\quad \psi_1=e^{(0)}+\frac{L}{2}e^{(1)}+\frac{L^2}{8}e^{(2)}+o(L^2), \\  
&e^{(0)}=[\psi(x_1),\psi(x_2),\cdots,\psi(x_W)]^T,\\
&e^{(1)}=[\psi^{(1)}(x_1),\psi^{(1)}(x_2),\cdots,\psi^{(1)}(x_W)]^T,\\
&e^{(2)}=[\psi^{(2)}(x_1),\psi^{(2)}(x_2),\cdots,\psi^{(2)}(x_W)]^T,
\end{aligned}
\end{equation}
where $\psi^{(1)}(x)=d\psi/dx$, $\psi^{(2)}(x)=d^2\psi/dx^2$.
\yk{For a Gaussian PSF $\psi(x) = (2\pi\sigma^2)^{-1/4}\exp(-x^2/(4\sigma^2))$ and taking $W \to \infty$ to enable the use of integrals in evaluating $|e^{(i)}|$, we obtain
\begin{equation}\begin{aligned}
&\|e^{(0)}\|^2=\int dx \psi(x)^2=1,\\
&\|e^{(1)}\|^2=\int dx\frac{x^2}{4\sigma^2} \psi(x)^2=\frac{1}{4\sigma^2},\\
& \|e^{(2)}\|^2=\int dx\frac{(x^2-2\sigma^2)^2}{16\sigma^8} \psi(x)^2=\frac{3}{16\sigma^4},
\end{aligned}
\end{equation}
However, we emphasize that the above proof is valid for any PSF shape; only the constant arising from $\|e^{(i)}\|$ depends on the specific form of the PSF.}
We expand $\frac{\partial\Gamma}{\partial L}$ as a series of $L$ as $L\rightarrow 0$
\begin{equation}
\frac{\partial\Gamma}{\partial L}=\frac{\epsilon}{2}L\left(e^{(1)}(e^{(1)})^\dagger+\frac{1}{2}e^{(0)}(e^{(2)})^\dagger+\frac{1}{2}e^{(2)}(e^{(0)})^\dagger\right)+\epsilon O(L^2).
\end{equation}
Note that as we take $W \rightarrow \infty$, only the dimension of $e^{(i)}$ in $\frac{\partial \Gamma}{\partial L}$ is affected, while $\|e^{(i)}\|^2$ remains finite. Consequently, the term $\epsilon O(L^2)$ does not diverge as $W \rightarrow \infty$.
Importantly, the $\Theta(L^0)$ order terms of $\frac{\partial\Gamma}{\partial L}$ vanishes.
We can easily find that
\begin{equation}
 \lambda_{\text{max}}\left(\left(\frac{\partial \Gamma}{\partial L}\right)^2\right)\leq \left(\frac{\epsilon}{2}L(||e^{(1)}||^2+  \|e^{(0)}\|\|e^{(2)}\|)\right)^2+\epsilon^2 O(L^3)=\frac{\epsilon^2L^2(\sqrt{3}+1)^2}{64\sigma^4}+\epsilon^2 O(L^3),
\end{equation}
where we treat $\epsilon$ as a constant without taking the limit $\epsilon\rightarrow 0$.
Similar to the proof of Eq.~\ref{eq:F_theta_single_lens}, we can have
\begin{equation}
F_{LL}\leq \frac{1}{2}\lambda_{\text{max}}\left(\left(\frac{\partial \Gamma}{\partial L}\right)^2\right)\tr((P(V^N)^{-1}P)^2),
\end{equation}
where $P$ is the projector onto the support of $\frac{\partial G^N}{\partial L}$, leading to $\text{rank}(P) = \text{rank}(\frac{\partial G^N}{\partial L}) \leq 8N$. Since a diagonal block of a positive semidefinite matrix remains positive semidefinite, it follows that
\begin{equation}
F_{LL}\leq \frac{1}{2}\lambda_{\text{max}}\left(\left(\frac{\partial \Gamma}{\partial L}\right)^2\right)\tr((P(V^N_\rho)^{-1}P)^2)\leq \frac{1}{2}\lambda_{\text{max}}\left(\left(\frac{\partial \Gamma}{\partial L}\right)^2\right)\text{rank}(\frac{\partial G^N}{\partial L})\lambda_{\text{max}}\left(V_\rho^{-2}\right),
\end{equation}
where, in the second inequality, we use the fact that $\lambda_{\text{max}}\left(P(V^N_\rho)^{-2}P\right) \leq \lambda_{\text{max}}\left(V_\rho^{-2}\right)$, similar to Eq.~\ref{eq:PVP}. Since $G$ is positive semidefinite, all of the eigenvalues of $(V_{\rho})^{-1}$ satisfy $\lambda_i((V_{\rho})^{-1}) \leq 2$.
Collecting all the bounds above, we have shown that
\yk{\begin{equation}\label{SI_eq:FI_L}
F_{LL}\leq4\epsilon^2L^2(||e^{(1)}||^2+  \|e^{(0)}\|\|e^{(2)}\|)^2+N\epsilon^2O(L^3)= \frac{N\epsilon^2L^2(\sqrt{3}+1)^2}{4\sigma^4}+N\epsilon^2O(L^3),
\end{equation}
where the equality is for the Gaussian PSF as an example. The bound for $F_{LL}$ confirms the $NO(\epsilon^2)$ scaling as claimed for imaging weak thermal sources using Gaussian measurements.} Furthermore, as $L \to 0$, we find that $F \to 0$, implying that Gaussian measurements cannot achieve $F_{LL}$ independent of the separation $L$ between two point sources. Consequently, superresolution, as demonstrated in Ref.~\cite{tsang2016quantum}, is not attainable with Gaussian measurements.
We note that $F_{LL}$ is not a function of $L/\sigma$ because the FI for estimating $L$ involves differentiation of a unitless probability with respect to $L$, which has units and is therefore not unitless. Consequently, the denominator of $F$ contains a factor of $\sigma^4$, while the numerator includes a factor of $L^2$. See, for example, Ref.~\cite{tsang2016quantum}, which also derives FIM expressions that are not functions of $L/\sigma$. \yk{We emphasize that the prefactor $(\sqrt{3}+1)^2/4$ of $F_{LL}$ in Eq.~\ref{SI_eq:FI_L} is determined by the Gaussian PSF considered here as an example; however, the no-go proof applies to any PSF, with only the constant prefactor $(||e^{(1)}||^2+  \|e^{(0)}\|\|e^{(2)}\|)^2$ changing, as it is determined by the given PSF.
}

We can also extend this no-go theorem for superresolution to the case of interferometric imaging. We refer the reader to Ref.~\cite{wang2021superresolution} for the derivation of the received state and the discussion using non-Gaussian measurement, which can achieve a FI that is independent of $L$ and is of order $N\epsilon$. We now prove that any Gaussian measurement can only achieve an FI of order $N\epsilon^2O(L^2)$. The covariance matrix of the received states by two lenses when imaging two thermal sources at positions $\pm L/2$ is given by
\begin{equation}
V_\rho=\frac{1}{2}I_4+G, \quad G=I_2\otimes\Gamma,\quad \Gamma=\frac{\epsilon}{2}\left[
\begin{matrix}
1 & \cos(kL)  \\   
\cos(kL) & 1\\
\end{matrix}\right],
\end{equation}
where $k$ is a constant depends on the distance between the two lenses. The FI of estimating $L$ is bounded by
\begin{equation}
F_{LL}=\frac{1}{2}\tr((V^N)^{-1}\frac{\partial V^N}{\partial L}(V^N)^{-1}\frac{\partial V^N}{\partial L})\leq\frac{1}{2}\lambda_{\text{max}}\left(\left(\frac{\partial \Gamma}{\partial L}\right)^2\right)\text{rank}(\frac{\partial G^N}{\partial L})\lambda_{\text{max}}\left(V_\rho^{-2}\right),
\end{equation}
where   $ \text{rank}(\frac{\partial G^N}{\partial L}) \leq 4N$, $\lambda_{\text{max}}\left(\left(\frac{\partial \Gamma}{\partial L}\right)^2\right)=\epsilon^2k^2\sin^2(kL)/4$, all of the eigenvalues of $(V_{\rho})^{-1}$ satisfy $\lambda_i((V_{\rho})^{-1}) \leq 2$. We thus have
\begin{equation}
F_{LL}\leq 2N\epsilon^2k^2\sin^2(kL)=N\epsilon^2O(L^2),
\end{equation}

\yk{
\subsection{Proof of Theorem \ref{thm:superreoslution} for the superresolution of general sources}

The state on the source plane, before entering the imaging system, is given by
\begin{equation}\begin{aligned}
&\rho=\int\frac{d^{2Q}\vec{\alpha} }{\pi^Q\det\Gamma_0}\exp(-\vec{\alpha}^\dagger \Gamma_0^{-1}\vec{\alpha})\ket{\vec{\alpha}}\bra{\vec{\alpha}},\\
&\Gamma_0=  \epsilon\,\text{diag}[\zeta_1,\zeta_2,\cdots,\zeta_Q],\quad\sum_{i=1}^Q\zeta_i=1\\
&\ket{\vec{\alpha}}=\prod_{i=1}^Q\exp(\alpha_i \hat{a}_i^\dagger-\alpha_i^*\hat{a}_i)\ket{0}.
\end{aligned}\end{equation}
where $Q\rightarrow\infty$ is the number of points on the source. The imaging system can cause the evolution
\begin{equation}
a_i\rightarrow \int dx \psi(x-y_i)c_x,
\end{equation}
where $c_x$ is the annihilation operator at position $x$ on the detection plane, $\psi(x)$ is the PSF, $y_i$ is the position of $i$th point on the source. Note that we assume $y_i\in[-L/2+y_0,L/2+y_0]$, where $L$ is the size of the source, $y_0$ is the centroid of the source. The state received at the detection plane is given by
\begin{equation}\begin{aligned}
&\rho=\int\frac{\prod_{i=1}^Wd^2\gamma_{x_i}}{\det(\pi\Gamma)}\exp(-\vec{\gamma}^\dagger \Gamma^{-1}\vec{\gamma})\ket{\vec{\gamma}}\bra{\vec{\gamma}},\\
&\vec{\gamma}=[\gamma_{x_1},\gamma_{x_2},\cdots,\gamma_{x_W}]^T,\quad \Gamma=R\Gamma_0R^\dagger=\sum_i\epsilon \zeta_i\psi_i\psi_i^\dagger, \\
&R=[\psi_1,\psi_2,\cdots,\psi_Q],\quad\psi_i=[\psi(x_1-y_i),\psi(x_2-y_i),\cdots,\psi(x_W-y_i)]^T,\\
&\ket{\vec{\gamma}}=\exp(\sum_i\gamma_{x_i}c_{x_i}^\dagger-\gamma_{x_i}^*c_{x_i})\ket{0},
\end{aligned}\end{equation}
Since $\Gamma$ is real for real PSF, the covariance matrix of the state is given by
\begin{equation}
V_\rho=\frac{1}{2}I_{2W}+G,\quad G=I_2\otimes \Gamma.
\end{equation}
The FIM element for estimating the separation $L$ from $\rho^{\otimes N}$ is given by (as before, we use $V^N, G^N, \dots$ to denote the tensor product with $I_N$ when considering $N$ copies of the state)
\begin{equation}
F_{t_nt_n}=\frac{1}{2}\tr((V^N)^{-1}\frac{\partial V^N}{\partial t_n}(V^N)\frac{\partial V^N}{\partial t_n})=\frac{1}{2}\tr((V^N)^{-1}\frac{\partial G^N}{\partial t_n}(V^N)^{-1}\frac{\partial G^N}{\partial t_n}),
\end{equation}
where $V^N=V^N_\rho+V^N_\Pi$, $V^N_\Pi$ is covariance matrix of the Gaussian measurement. 
We use the expansion below
\begin{equation}
\begin{aligned}
&\psi_i=\sum_{n=0}^\infty\left(\frac{y_i-y_0}{L}\right)^n\omega^{(n)} ,\quad \omega^{(n)}=\frac{L^n}{n!}\left[\frac{\partial^n \psi(x_1-y)}{\partial y^n}\bigg|_{y=y_0},\frac{\partial^n \psi(x_2-y)}{\partial y^n}\bigg|_{y=y_0},\cdots, \frac{\partial^n \psi(x_W-y)}{\partial y^n}\bigg|_{y=y_0}\right]^T, \\
&\Gamma=\sum_{n_1,n_2=0}^\infty\epsilon \, t_{n_1+n_2}\,\omega^{(n_1)}(\omega^{(n_2)})^T,\quad t_n:=\sum_{i=1}^Q\zeta_i\left(\frac{y_i-y_0}{L}\right)^n
\end{aligned}
\end{equation}
where $t_n$ is the normalized moments. Note that $\omega^{(n)}\propto L^n$. In the small source limit, which is relevant for the superresolution discussion, we have $L\rightarrow 0$. Thus, we have get an expansion of $\Gamma$ as a series of $L$.  We can then take the derivative
\begin{equation}
\frac{\partial \Gamma}{\partial t_n}=\sum_{n_1+n_2=n}\epsilon\omega^{(n_1)}(\omega^{(n_2)})^T=\Theta(L^n)
\end{equation}
We can thus bound the eigenvalue 
\begin{equation}\begin{aligned}
&\lambda_{\text{max}}\left(\frac{\partial \Gamma}{\partial t_n}\right)\leq \sum_{n_1+n_2=n}\epsilon\|\omega^{(n_1)}\|\|\omega^{(n_2)}\|\\
&\|\omega^{(n)}\|^2=\frac{L^{2n}}{(n!)^2}\int dx \left(\frac{\partial^n \psi(x-y)}{\partial y^n}\bigg|_{y=y_0}\right)^2\propto L^{2n}\\
\end{aligned}\end{equation}
For the case of Gaussian PSF $\psi(x)=(2\pi\sigma^2)^{-1/4}\exp(-x^2/(4\sigma^2))$, we can easily calculate the value of each $\|\omega^{(n)}\|$
\begin{equation}
\|\omega^{(0)}\|^2=1,\quad \|\omega^{(1)}\|^2=\frac{L^2}{4\sigma^2},\quad \|\omega^{(2)}\|^2=\frac{3L^4}{64\sigma^4},\cdots
\end{equation}
We emphasize that our proof applies to any PSF, and the specific shape of the PSF only affects the constant factors arising from $\|\omega^{(n)}\|^2$.
Similar to the proof of Eq.~\ref{eq:F_theta_single_lens}, we can have
\begin{equation}
F_{t_nt_n}\leq \frac{1}{2}\lambda_{\text{max}}\left(\left(\frac{\partial \Gamma}{\partial t_n}\right)^2\right)\tr((P(V^N)^{-1}P)^2),
\end{equation}
where $P$ is the projector onto the support of $\frac{\partial G^N}{\partial t_n}$, leading to $\text{rank}(P) = \text{rank}(\frac{\partial G^N}{\partial t_n}) \leq 2N(n+1)$. Since a diagonal block of a positive semidefinite matrix remains positive semidefinite, it follows that
\begin{equation}
F_{t_nt_n}\leq \frac{1}{2}\lambda_{\text{max}}\left(\left(\frac{\partial \Gamma}{\partial t_n}\right)^2\right)\tr((P(V^N_\rho)^{-1}P)^2)\leq \frac{1}{2}\lambda_{\text{max}}\left(\left(\frac{\partial \Gamma}{\partial t_n}\right)^2\right)\text{rank}(\frac{\partial G^N}{\partial t_n})\lambda_{\text{max}}\left(V_\rho^{-2}\right),
\end{equation}
where, in the second inequality, we use the fact that $\lambda_{\text{max}}\left(P(V^N_\rho)^{-2}P\right) \leq \lambda_{\text{max}}\left(V_\rho^{-2}\right)$, similar to Eq.~\ref{eq:PVP}. Since $G$ is positive semidefinite, all of the eigenvalues of $(V_{\rho})^{-1}$ satisfy $\lambda_i((V_{\rho})^{-1}) \leq 2$.
Collecting all the bounds above, we have shown that
\begin{equation}%\label{SI_eq:FI_L}
F_{t_nt_n}\leq 4N(n+1)\epsilon^2\left(\sum_{n_1+n_2=n}\|\omega^{(n_1)}\|\|\omega^{(n_2)}\|\right)^2=\Theta(N\epsilon^2L^{2n})
\end{equation}
Then, since the FIM is a positive semidefinite matrix, we can bound each element of the FIM as
\begin{equation}
F_{t_nt_m}\leq \sqrt{F_{t_nt_n}F_{t_mt_m}}=\Theta(N\epsilon^2L^{n+m})
\end{equation}

}

\yk{
\section{Unbiased estimator}\label{Appendix:Unbiased estimator}

The standard Cramer-Rao bound constrains only the precision of unbiased estimators. And biased estimators may potentially violate the standard Cramer-Rao bound, representing a possible loophole. Despite this limitation, the standard Cramer-Rao bound remains a widely used tool in the majority of quantum imaging studies and appears to yield reliable results in most cases. The discussion on this issue for the superresolution problem is given in the simple case of resolving the separation $L$ between two point source \cite{tsang2018conservative}. They show that even including the biased estimator, the superresolution approach using non-Gaussian measurement could still show better performance than the direct imaging. And here, we further show that even with biased estimator, the Gaussian measurement always has only comparable performance to the direct imaging. We first review their results and then present our results.

\subsection{Review of Ref.~\cite{tsang2018conservative}}

%To avoid the problem of not including the biased  estimator in the standard Cramer-Rao bound, they use the Bayesian Cramer-Rao bound \cite{trees2007bayesian,gill1995applications,schutzenberger1957generalization,van2004detection}, which includes both the biased and unbiased estimators. The Bayesian Cramer-Rao bound is defined to bound the mean square error of estimating the unknown parameter $L$ given the prior information $p(L)$. To have better generality, they choose the figure of merit, worse-case error  $\sup_L\text{MSE} (L)$, which is largest mean square error for estimating $L$ optimized over all possible prior $p(L)$.

%To avoid the problem of not including the biased  estimator in the standard Cramer-Rao bound, they use the figure of merit, worse-case error  $\sup_L\text{MSE (L)}$, which is largest mean square error for estimating $L$ for all the values of $L$. Note that typically the mean square error of estimating $L$ depends on the value of $L$, so the worse-case error here is defined. Then,  for any prior distribution of $p(L)$, we have

%To avoid the problem of not including the biased  estimator in the standard Cramer-Rao bound, they use the figure of merit, worse-case error  $\sup_L\text{MSE (L)}$, which is largest mean square error for estimating $L$ for all the values of $L$. To bound the worse-case error, they use the Bayesian Cramer-Rao bound, which is a generalization of the standard Cramer-Rao bound and incorporates the prior information to bound the mean square error. For any prior distribution of $p(L)$, we have

We summarize the results in Ref. \cite{tsang2018conservative} below with our notation convention. To address the limitation of the standard Cramer-Rao bound, which excludes biased estimators, they use the Bayesian Cramer-Rao bound \cite{trees2007bayesian,gill1995applications,schutzenberger1957generalization,van2004detection}, which applies to both biased and unbiased estimators. The Bayesian Cramer-Rao bound constrains the mean square error $\int dL \, p(L)\,\text{MSE}(L)$ for estimating the unknown parameter $L$ given the prior distribution $p(L)$. They introduce the worst-case error $\sup_L \text{MSE}(L)$, defined as the maximum mean square error for estimating $L$ across all possible values of $L$, which can be bounded by the mean square error $\int dL \, p(L)\,\text{MSE}(L)$.
\begin{equation}\label{SI_eq:K_p}
\begin{aligned}
&\sup_L\text{MSE} (L)\geq \int dL \,p(L)\,\text{MSE} (L)\geq \frac{1}{K[p(L)]},\\
&K[p(L)]=\int dL p(L) F(L)+j[p(L)],\\
&j[p(L)]=\int dL p(L)\left[\frac{\partial \ln p(L)}{\partial L}\right]^2,
\end{aligned}
\end{equation}
where $F(L)$ is the Fisher information of estimating $L$ for the SPADE or direct imaging approach at value $L$ of the separation, the first inequality holds by definition since the worse-case error is no smaller than the average error, the second inequality is due to the Bayesian Cramer-Rao bound \cite{trees2007bayesian,gill1995applications,schutzenberger1957generalization,van2004detection}, which is valid for any estimator.
Since the first in equality holds for any $p(L)$, to get a better bound for the worse-case error  $\sup_L\text{MSE} (L)$, they minimize the $K[p(L)]$ by choosing $p(L)$.

The minimization is performed under the constraint that the prior probability is normalized, $\int dL\, p(L) = 1$. This is enforced using the standard Lagrange multiplier method for constrained optimization. For convenience, they define $q(L)$ such that $p(L) = q^2(L)$. 
With $q(L)$, the quantity $K[p(L)] = K[q(L)] = \int dL \,[ q^2(L) F(L) + 4 (q'(L))^2 ]$, where $q'(L) = dq(L)/dL$ (and similarly for other functions of $L$), is subject to the normalization constraint $\int dL\, q^2(L) = 1$. To enforce this constraint,  introduce a Lagrange multiplier $\lambda$ and define
\begin{equation}
\tilde{K}[q(L)] = \int dL \,\left[ q^2(L) F(L) + 4\,(q'(L))^2 - \lambda\,(q^2(L) - 1) \right].
\end{equation}
Varying $q(L) \to q(L) + \epsilon\,\eta(L)$ and keeping $O(\epsilon)$ terms gives
\begin{equation}
\delta \tilde{K}[q(L)]= \int dL \,\left[ 2q(L)\,F(L)\,\eta(L) + 8\,q'(L)\,\eta'(L) - 2\lambda\,q(L)\,\eta(L) \right].
\end{equation}
Integrating the $q'(L)\eta'(L)$ term by parts and assuming vanishing boundary terms,
\begin{equation}
\delta \tilde{K}[q(L)] = \int dL \,\left[ 2q(L)\,F(L) - 2\lambda\,q(L) - 8\,q''(L) \right] \eta(L), \quad q''(L)=\frac{d^2 q(L)}{dL^2}.
\end{equation}
Since $\eta(L)$ is arbitrary, the condition $\delta\tilde{K}[q(L)]=0$ yields
\begin{equation}\label{SI_eq:q_L}
\lambda q(L)=-4\frac{d^2q(L)}{dL^2}+F(L)q(L)
\end{equation}
which forms an eigenvalue problem, whose solution yields the Lagrange multiplier $\lambda$ and the optimal $q(L)$—and thus the optimal prior $p(L)$. If multiple solutions exist, each solution might yield a looser upper bound for $K[p(L)]$ or, equivalently, a looser lower bound for the worst-case error $\sup_L \text{MSE}(L)$.

In the problem of resolving two point sources of equal strength, they determine $p(L)$ by solving Eq.~\ref{SI_eq:q_L} with the corresponding Fisher information $F(L)$. They find that the SPADE method (which using non-Gaussian measurement) has
\begin{equation}\label{SI_eq:K_p1}
K[p(L)]\leq \frac{N}{4\sigma^2}
\end{equation}
whereas for the direct imaging method, it is possible to find $p(L)$ such that
\begin{equation}\label{SI_eq:K_p2}
K[p(L)]\leq \frac{3\sqrt{N}}{\sqrt{2}\sigma^2}
\end{equation}
where $N$ is the number of photons detected on the image plane. So, they show a scaling difference in terms of $N$ for direct imaging and the SPADE method. And it is in this sense that the SPADE method still outperform direct imaging even if we include the biased estimator. The analytical results obtained in their work are also validated by numerical simulations, confirming the predicted scaling difference in $N$ between direct imaging and superresolution. Note that the bound for the worst-case error using the Bayesian Cramér–Rao bound is still based on the Fisher information. Eqs. \ref{SI_eq:K_p1} and \ref{SI_eq:K_p2} are determined by the behavior of the Fisher information $F(L)$ when solving Eq. \ref{SI_eq:q_L}, as detailed in Ref.~\cite{tsang2018conservative}.

\subsection{Non-Gaussian measurement case}

Once deriving the bound of Fisher information $F(L)$ for resolving two point source using any Gaussian measurement as in Eq. \ref{SI_eq:FI_L}, $F_{LL}\leq k{NL^2}/{\sigma^4}+N\epsilon^2O(L^3)$, $k=\epsilon^2(\sqrt{3}+1)^2/4$, we can easily follow the approach in Ref. \cite{tsang2018conservative} to show any Gaussian measurement also has the scaling $1/\sqrt{N}$. We now give the detailed proof. In the limit $L \rightarrow 0$, substituting the leading-order term of the upper bound on $F_{LL}$ into Eq.~\ref{SI_eq:q_L} yields $q(L)$ and $\lambda$ as
\begin{equation}
q(L)=\left(\frac{2}{\pi}\right)^{1/4}\frac{L}{\omega^{3/2}}\exp\left(-\frac{L^2}{4\omega^2}\right),\quad\omega=\frac{\sigma}{(kN)^{1/4}},\quad \lambda=6/\omega^2,
\end{equation}
Using the prior distribution solved here $p(L)=q^2(L)$, and plug into Eq. \ref{SI_eq:K_p}, we can find
\begin{equation}\label{SI_eq:K_pL}
K[p(L)]\leq\frac{6\sqrt{kN}}{\sigma^2},
\end{equation}
which shows that for any Gaussian measurement, the scaling for worse-case error  $\sup_L\text{MSE} (L)$ scales as $\sqrt{N}$ which is comparable to the direct imaging approach. 
We emphasize that the bound in Eq. \ref{SI_eq:K_pL} applies to any Gaussian measurement because it follows directly from the Fisher information bound in Eq. \ref{SI_eq:FI_L}, which holds for all Gaussian measurements.  Therefore, we close the loophole of potential advantage provided by biased estimators by showing that for any estimator the Gaussian measurement cannot achieve superresolution.   
}

\end{document}